\documentclass[journal]{IEEEtran}

\usepackage[utf8]{inputenc}
\usepackage{color}
\usepackage{xcolor}
\usepackage{array}
\usepackage{verbatim}
\usepackage{float}
\usepackage{amsmath}
\usepackage{amsthm}
\usepackage{amssymb}
\usepackage{graphicx}
\usepackage{longtable}
\usepackage{multirow}
\usepackage[unicode=true,
bookmarks=false,
breaklinks=false,pdfborder={0 0 1},colorlinks=false]
{hyperref}
\hypersetup{
	colorlinks,bookmarksopen,bookmarksnumbered,citecolor=blue,urlcolor=blue}
\usepackage{cite}

\usepackage{lipsum}
\usepackage{mathtools}
\usepackage{cuted}

\floatstyle{ruled}
\newfloat{algorithm}{tbp}{loa}
\providecommand{\algorithmname}{Algorithm}
\floatname{algorithm}{\protect\algorithmname}

\makeatletter
\let\oldforeign@language\foreign@language
\DeclareRobustCommand{\foreign@language}[1]{%
	\lowercase{\oldforeign@language{#1}}}

\let\oldforeign@language\foreign@language
\DeclareRobustCommand{\foreign@language}[1]{%
	\lowercase{\oldforeign@language{#1}}}

\ifCLASSINFOpdf
\else
\fi

\newcommand{\MYfooter}{\smash{
		\hfil\parbox[t][\height][t]{\textwidth}{\centering
			\thepage}\hfil\hbox{}}}

\def\ps@IEEEtitlepagestyle{%
	\def\@oddhead{\parbox[t][\height][t]{\textwidth}{\centering \scriptsize
			Personal use of this material is permitted. Permission from the author(s) and/or copyright holder(s), must be obtained for all other uses. Please contact us and provide details if you believe this document breaches copyrights.\\
			\noindent\makebox[\linewidth]{}
		}\hfil\hbox{}}%
	\def\@evenhead{\scriptsize\thepage \hfil \leftmark\mbox{}}%
	\def\@oddfoot{\parbox[t][\height][l]{\textwidth}{
			\vspace{-20pt}{\rule{\textwidth}{0.4pt}}\\ \footnotesize\underline{To cite this article:}
			{\bf{\footnotesize\textcolor{red}{H. A. Hashim "A Geometric Nonlinear Stochastic Filter for Simultaneous Localization and Mapping," Aerospace Science and Technology, vol. PP, no. PP, pp. PP, 2021.}}} doi: \href{https://doi.org/10.1016/j.ast.2021.106569}{10.1016/j.ast.2021.106569}\\
			\noindent\makebox[\linewidth]
		}\hfil\hbox{}}%
	\def\@evenfoot{\MYfooter}}

\makeatother
\newtheorem{defn}{Definition}

\newtheorem{lem}{Lemma}

\newtheorem{thm}{Theorem}
\newtheorem{rem}{Remark}

\newtheorem{assum}{Assumption}

\begin{document}
	\bstctlcite{IEEEexample:BSTcontrol}

	\title{A Geometric Nonlinear Stochastic Filter for Simultaneous Localization and Mapping}

\author{Hashim A. Hashim$^*$\IEEEmembership{~Member, IEEE}
	\thanks{This work was supported in part by Thompson Rivers University Internal research fund \# 102315.}
	\thanks{$^*$Corresponding author, H. A. Hashim is with the Department of Engineering and Applied Science, Thompson Rivers University, Kamloops, British Columbia, Canada, V2C-0C8, e-mail: hhashim@tru.ca}
}



\maketitle

\begin{abstract}
Simultaneous Localization and Mapping (SLAM) is one of the key robotics
tasks as it tackles simultaneous mapping of the unknown environment
defined by multiple landmark positions and localization of the unknown
pose (\textit{i.e}., attitude and position) of the robot in three-dimensional
(3D) space. The true SLAM problem is modeled on the Lie group of $\mathbb{SLAM}_{n}\left(3\right)$,
and its true dynamics rely on angular and translational velocities.
This paper proposes a novel geometric nonlinear stochastic estimator
algorithm for SLAM on $\mathbb{SLAM}_{n}\left(3\right)$ that precisely
mimics the nonlinear motion dynamics of the true SLAM problem. Unlike
existing solutions, the proposed stochastic filter takes into account
unknown constant bias and noise attached to the velocity measurements.
The proposed nonlinear stochastic estimator on manifold is guaranteed
to produce good results provided with the measurements of angular
velocities, translational velocities, landmarks, and inertial measurement
unit (IMU). Simulation and experimental results reflect the ability
of the proposed filter to successfully estimate the six-degrees-of-freedom
(6 DoF) robot's pose and landmark positions. 
\end{abstract}

\begin{IEEEkeywords}
Simultaneous Localization and Mapping, nonlinear stochastic observer,
stochastic differential equations, pose estimator, position, attitude,
Brownian motion process, inertial measurement unit, SLAM.
\end{IEEEkeywords}

\IEEEpeerreviewmaketitle{}

\section{Introduction}

\IEEEPARstart{R}{obotics} applications are experiencing a surge in demand for navigation
solutions suitable for partially or completely unknown robot pose
in three-dimensional (3D) space (\textit{i.e}., attitude and position)
within an unknown environment. Robot's pose is comprised of two elements:
robot's orientation, also known as attitude, and robot's position.
Estimating map of the environment given robot's pose constitutes a
mapping problem popular within computer science and robotics communities
\cite{thrun2002robotic}. On the other hand, recovering robot's pose
within a known environment is referred to as pose estimation problem
long-established and well-detailed among robotics and control community
\cite{hashim2019SE3Det,zlotnik2018higher}. When neither the robot's
pose nor the map of the environment are known, the problem is termed
Simultaneous Localization and Mapping (SLAM). SLAM concurrently maps
the environment and localizes the robot with respect to the map. Unreliability
of absolute positioning systems, such as global positioning systems,
in occluded environments makes SLAM indispensable for a number of
applications, such as terrain mapping, multipurpose household robots,
mine exploration, locating missing terrestrial objects, reef monitoring,
surveillance, and others. Thus, for over a decade, SLAM and SLAM-related
applications have been a fundamental and widely-explored problem \cite{guo2020real,durrant2006simultaneous,sazdovski2015implicit,hashim2020LetterSLAM,zlotnik2018SLAM,li2018autonomous,milford2008mapping,sim2007study}.

The SLAM problem is traditionally addressed employing the measurements
available in the body-frame of a moving robot. Due to the fact that
measurements are contaminated with uncertain elements, SLAM estimation
requires a robust filter. The problem of SLAM estimation is conventionally
tackled using Gaussian filters or nonlinear deterministic filters.
Over ten years ago, several Gaussian filters for SLAM tailored specifically
to the task of estimating the robot state and the surrounding landmarks
were proposed. Gaussian solutions include MonoSLAM using real-time
single camera \cite{milford2008mapping}, FastSLAM using scalable
approach \cite{eade2006scalable}, incremental SLAM \cite{kaess2011isam2},
unscented Kalman filter (UKF) \cite{huang2013quadratic}, particle
filter \cite{sim2007study}, invariant EKF \cite{barczyk2014experimental},
in addition to others. These solutions account for uncertainties and
rely on probabilistic framework. SLAM problem presents a number of
open challenges, namely consistency \cite{dissanayake2011review},
computational cost and solution complexity \cite{cadena2016past},
as well as landmarks in motion. Other significant challenges that
hinder SLAM estimation process are as follows: 1) complexity of simultaneous
localization and mapping further complicated by 3D motion of the robot,
2) duality of the problem that requires simultaneous pose and map
estimation, and most importantly 3) high nonlinearity of the SLAM
problem. To address nonlinearity, it is important to note that true
motion dynamics of SLAM are composed of robot's pose dynamics and
landmark dynamics. Firstly, the highly nonlinear pose dynamics of
a robot traveling in 3D space are modeled on the Lie group of the
special Euclidean group $\mathbb{SE}\left(3\right)$. Secondly, robot's
attitude is an essential part of the landmark dynamics, and therefore
the attitude is described according to the Special Orthogonal Group
$\mathbb{SO}\left(3\right)$. Thereby, the key to successfully SLAM
estimation lies in utilizing filter design that captures the true
nonlinear structure of the problem.

Novel nonlinear attitude and pose filters evolved on $\mathbb{SO}\left(3\right)$
\cite{hashim2018SO3Stochastic,hashim2019SO3Wiley,jensen2011generalized,zamani2013minimum}
and on $\mathbb{SE}\left(3\right)$ \cite{hashim2019SE3Det,zlotnik2018higher,hashim2020SE3Stochastic}
enabled the development of nonlinear filters for SLAM. The fact that
nonlinear attitude and pose filters mimic the true attitude and pose
dynamics, served as a motivation for adopting the Lie group of $\mathbb{SE}\left(3\right)$
in application to the SLAM problem \cite{strasdat2012local}. A dual
nonlinear filter comprised of a nonlinear filter for robot's pose
estimation and a Kalman filter for landmark observation was proposed
\cite{johansen2016globally}. However, nonlinear nature of the true
SLAM problem was not yet completely captured by the work in \cite{johansen2016globally}.
As a result, nonlinear filters for SLAM that use measurements of landmarks
and group velocity vectors directly were developed \cite{hashim2020LetterSLAM,zlotnik2018SLAM,hashim2020TITS_SLAM}.
The work in \cite{hashim2020LetterSLAM,zlotnik2018SLAM,hashim2020TITS_SLAM,hashim2021T_SMCS_SLAM}
considered unknown constant bias inherent in the group velocity vector
measurements.

To this end, two major challenges must be considered during the design
process of a nonlinear filter for SLAM: 1) the SLAM problem is modeled
on the Lie group of $\mathbb{SLAM}_{n}\left(3\right)$ which is highly
nonlinear; and 2) the true SLAM kinematics rely on a group of velocities,
namely angular velocity, translational velocity, and velocities of
landmarks expressed relative to the body-frame. As such, successful
estimation can be attained by designing a nonlinear filter that relies
on the previously mentioned group of velocities which are normally
corrupted with unknown noise as well as unknown constant bias components.
Moreover, noise components are distinguished by random behavior, and
it is well recognized that noise can negatively impact the output
performance \cite{stojanovic_P1,hashim2018SO3Stochastic,hashim2020SE3Stochastic}.
To the best of the author knowledge, SLAM estimation problem has been
neither addressed nor solved in stochastic sense. As a result, it
is important to take into account any noise and/or bias components
present in the measurement process. Having this in mind and given
the following set of available measurements: group velocity vector,
$n$ landmarks and an inertial measurement unit (IMU), this paper
introduces a novel nonlinear stochastic filter for SLAM that has the
structure of nonlinear deterministic filters adapting it to the stochastic
sense. The main contributions of this paper are listed below: 
\begin{enumerate}
	\item[1)] A geometric nonlinear stochastic filter for SLAM developed directly
	on the Lie group of $\mathbb{SLAM}_{n}\left(3\right)$ which exactly
	follows the nonlinear structure of the true SLAM problem is proposed.
	\item[2)] The proposed nonlinear stochastic filter accounts for unknown constant
	bias and random noise attached to the group velocity measurements,
	unlike \cite{hashim2020LetterSLAM,zlotnik2018SLAM}.
	\item[3)] The closed loop error signals of the Lyapunov candidate function
	are shown to be semi-globally uniformly ultimately bounded (SGUUB)
	in mean square.
	\item[4)] The proposed stochastic filter involves gain mapping that takes into
	account cross coupling between the innovation of pose and landmarks.
\end{enumerate}
The rest of the paper is structured in the following manner: Section
\ref{sec:Preliminaries-and-Math} contains an overview of the preliminaries
as well as introduces mathematical notation, the Lie group of $\mathbb{SO}\left(3\right)$,
$\mathbb{SE}\left(3\right)$, and $\mathbb{SLAM}_{n}\left(3\right)$.
Section \ref{sec:SE3_Problem-Formulation} describes the SLAM problem,
true motion kinematics and formulates the SLAM problem in a stochastic
sense. Section \ref{sec:SLAM_Filter} outlines a common structure
of nonlinear deterministic filter for SLAM on $\mathbb{SLAM}_{n}\left(3\right)$
and then proposes a nonlinear stochastic filter design on $\mathbb{SLAM}_{n}\left(3\right)$.
Section \ref{sec:SE3_Simulations} shows the effectiveness of the
proposed stochastic filter. Finally, Section \ref{sec:SE3_Conclusion}
concludes the work.

\section{Math Notation and $\mathbb{SLAM}_{n}\left(3\right)$ Preliminaries
	\label{sec:Preliminaries-and-Math}}

Throughout this paper two frames of reference are used: $\left\{ \mathcal{I}\right\} $
is a fixed inertial frame and $\left\{ \mathcal{B}\right\} $ is a
moving body-frame of a robot. $\mathbb{R}$, $\mathbb{R}_{+}$, and
$\mathbb{R}^{p\times q}$ denote sets of real numbers, nonnegative
real numbers, and a real space of dimension $p$-by-$q$, respectively.
$\mathbf{I}_{n}$ represents an identity matrix with dimension $n$,
$\underline{\mathbf{0}}_{n}$ represents a vector comprised of zeros,
and $\left\Vert y\right\Vert =\sqrt{y^{\top}y}$ represents Euclidean
norm for $y\in\mathbb{R}^{n}$. $\mathbb{P}\left\{ \cdot\right\} $,
$\mathbb{E}\left[\cdot\right]$, and ${\rm exp}\left(\cdot\right)$
denote probability, an expected value, and an exponential of a component,
respectively. $\mathcal{C}^{n}$ stands for a set of functions characterized
by continuous $n$th partial derivatives. $\mathcal{K}_{\infty}$
represents a set of functions whose elements are continuous and strictly
increasing. The Special Orthogonal Group $\mathbb{SO}\left(3\right)$
is expressed as 
\[
\mathbb{SO}\left(3\right)=\left\{ \left.R\in\mathbb{R}^{3\times3}\right|RR^{\top}=R^{\top}R=\mathbf{I}_{3}\text{, }{\rm det}\left(R\right)=+1\right\} 
\]
with ${\rm det\left(\cdot\right)}$ referring to a determinant of
a matrix, and $R\in\mathbb{SO}\left(3\right)$ being rigid-body's
orientation described in $\left\{ \mathcal{B}\right\} $, also known
as attitude. The Special Euclidean Group $\mathbb{SE}\left(3\right)$
is represented by
\[
\mathbb{SE}\left(3\right)=\left\{ \left.\boldsymbol{T}=\left[\begin{array}{cc}
R & P\\
\underline{\mathbf{0}}_{3}^{\top} & 1
\end{array}\right]\in\mathbb{R}^{4\times4}\right|R\in\mathbb{SO}\left(3\right),P\in\mathbb{R}^{3}\right\} 
\]
with $P\in\mathbb{R}^{3}$ being rigid-body's position and $R\in\mathbb{SO}\left(3\right)$
its orientation. $\boldsymbol{T}$ is used to express the rigid-body's
pose in 3D space and it is often referred to as a homogeneous transformation
matrix:
\begin{equation}
\boldsymbol{T}=\left[\begin{array}{cc}
R & P\\
\underline{\mathbf{0}}_{3}^{\top} & 1
\end{array}\right]\in\mathbb{SE}\left(3\right)\label{eq:T_SLAM}
\end{equation}
where $\underline{\mathbf{0}}_{3}$ denotes a zero column vector.
$\mathfrak{so}\left(3\right)$ represents the Lie-algebra associated
with $\mathbb{SO}\left(3\right)$ with
\begin{equation}
\mathfrak{so}\left(3\right)=\left\{ \left.\left[h\right]_{\times}\in\mathbb{R}^{3\times3}\right|\left[h\right]_{\times}^{\top}=-\left[h\right]_{\times},h\in\mathbb{R}^{3}\right\} \label{eq:SLAM_so3}
\end{equation}
such that $\left[h\right]_{\times}$ stands for a skew symmetric matrix.
The related map of \eqref{eq:SLAM_so3} $\left[\cdot\right]_{\times}:\mathbb{R}^{3}\rightarrow\mathfrak{so}\left(3\right)$
is
\[
\left[h\right]_{\times}=\left[\begin{array}{ccc}
0 & -h_{3} & h_{2}\\
h_{3} & 0 & -h_{1}\\
-h_{2} & h_{1} & 0
\end{array}\right]\in\mathfrak{so}\left(3\right),\hspace{1em}h=\left[\begin{array}{c}
h_{1}\\
h_{2}\\
h_{3}
\end{array}\right]
\]
and $\left[y\right]_{\times}h=y\times h$ where $\times$ represents
a cross product for $h,y\in\mathbb{R}^{3}$. $\mathfrak{se}\left(3\right)$
is the Lie-algebra associated with $\mathbb{SE}\left(3\right)$ defined
as 
\[
\mathfrak{se}\left(3\right)=\left\{ \left[U\right]_{\wedge}\in\mathbb{R}^{4\times4}\left|\exists\Omega,V\in\mathbb{R}^{3}:\left[U\right]_{\wedge}=\left[\begin{array}{cc}
\left[\Omega\right]_{\times} & V\\
\underline{\mathbf{0}}_{3}^{\top} & 0
\end{array}\right]\right.\right\} 
\]
with $\left[\cdot\right]_{\wedge}$ being a wedge operator. The related
wedge map $\left[\cdot\right]_{\wedge}:\mathbb{R}^{6}\rightarrow\mathfrak{se}\left(3\right)$
is defined by
\begin{equation}
\left[U\right]_{\wedge}=\left[\begin{array}{cc}
\left[\Omega\right]_{\times} & V\\
\underline{\mathbf{0}}_{3}^{\top} & 0
\end{array}\right]\in\mathfrak{se}\left(3\right),\hspace{1em}U=\left[\begin{array}{c}
\Omega\\
V
\end{array}\right]\in\mathbb{R}^{6}\label{eq:SLAM_wedge}
\end{equation}
The inverse mapping of $\left[\cdot\right]_{\times}$ is given by
$\mathbf{vex}:\mathfrak{so}\left(3\right)\rightarrow\mathbb{R}^{3}$
such that
\begin{equation}
\mathbf{vex}\left(\left[h\right]_{\times}\right)=h,\hspace{1em}\forall h\in\mathbb{R}^{3}\label{eq:SLAM_VEX}
\end{equation}
Consider $\boldsymbol{\mathcal{P}}_{a}$ to be an anti-symmetric projection
on $\mathfrak{so}\left(3\right)$ 
\begin{equation}
\boldsymbol{\mathcal{P}}_{a}\left(H\right)=\frac{1}{2}\left(H-H^{\top}\right)\in\mathfrak{so}\left(3\right),\hspace{1em}\forall H\in\mathbb{R}^{3\times3}\label{eq:SLAM_Pa}
\end{equation}
Let $\boldsymbol{\Upsilon}\left(\cdot\right)$ denote a composition
mapping of $\boldsymbol{\Upsilon}=\mathbf{vex}\circ\boldsymbol{\mathcal{P}}_{a}$
where
\begin{equation}
\boldsymbol{\Upsilon}\left(H\right)=\mathbf{vex}\left(\boldsymbol{\mathcal{P}}_{a}\left(H\right)\right)\in\mathbb{R}^{3},\hspace{1em}\forall H\in\mathbb{R}^{3\times3}\label{eq:SLAM_VEX_a}
\end{equation}
Define $\left\Vert R\right\Vert _{{\rm I}}$ as a normalized Euclidean
distance of $R\in\mathbb{SO}\left(3\right)$ with
\begin{equation}
\left\Vert R\right\Vert _{{\rm I}}=\frac{1}{4}{\rm Tr}\left\{ \mathbf{I}_{3}-R\right\} \in\left[0,1\right]\label{eq:SLAM_Ecul_Dist}
\end{equation}
Define $\overset{\circ}{\mathcal{M}}$ and $\overline{\mathcal{M}}$
as submanifolds of $\mathbb{R}^{4}$ 
\begin{align*}
\overset{\circ}{\mathcal{M}} & =\left\{ \left.\overset{\circ}{x}=\left[\begin{array}{cc}
x^{\top} & 0\end{array}\right]^{\top}\in\mathbb{R}^{4}\right|x\in\mathbb{R}^{3}\right\} \\
\overline{\mathcal{M}} & =\left\{ \left.\overline{x}=\left[\begin{array}{cc}
x^{\top} & 1\end{array}\right]^{\top}\in\mathbb{R}^{4}\right|x\in\mathbb{R}^{3}\right\} 
\end{align*}
Let $\mathbb{SLAM}_{n}\left(3\right)=\mathbb{SE}\left(3\right)\times\overline{\mathcal{M}}^{n}$
be a Lie group
\begin{equation}
\mathbb{SLAM}_{n}\left(3\right)=\left\{ X=\left(\boldsymbol{T},\overline{{\rm p}}\right)\left|\boldsymbol{T}\in\mathbb{SE}\left(3\right),\overline{{\rm p}}\in\overline{\mathcal{M}}^{n}\right.\right\} \label{eq:SLAM_SLAM_X}
\end{equation}
with $\overline{{\rm p}}=\left[\overline{{\rm p}}_{1},\overline{{\rm p}}_{2},\ldots,\overline{{\rm p}}_{n}\right]\in\overline{\mathcal{M}}^{n}$
and $\overline{\mathcal{M}}^{n}=\overline{\mathcal{M}}\times\overline{\mathcal{M}}\times\cdots\times\overline{\mathcal{M}}$.
$\mathfrak{slam}_{n}\left(3\right)=\mathfrak{se}\left(3\right)\times\overset{\circ}{\mathcal{M}}^{n}$
denotes a tangent space at the identity element of $X=\left(\boldsymbol{T},\overline{{\rm p}}\right)\in\mathbb{SLAM}_{n}\left(3\right)$
represented as
\begin{equation}
\mathfrak{slam}_{n}\left(3\right)=\left\{ \mathcal{Y}=\left(\left[U\right]_{\wedge},\overset{\circ}{{\rm v}}\right)\left|\left[U\right]_{\wedge}\in\mathfrak{se}\left(3\right),\overset{\circ}{{\rm v}}\in\overset{\circ}{\mathcal{M}}^{n}\right.\right\} \label{eq:SLAM_SLAM_Y}
\end{equation}
with $\overset{\circ}{{\rm v}}=\left[\overset{\circ}{{\rm v}}_{1},\overset{\circ}{{\rm v}}_{2},\ldots,\overset{\circ}{{\rm v}}_{n}\right]\in\overset{\circ}{\mathcal{M}}^{n}$
and $\overset{\circ}{\mathcal{M}}^{n}=\overset{\circ}{\mathcal{M}}\times\overset{\circ}{\mathcal{M}}\times\cdots\times\overset{\circ}{\mathcal{M}}$.
The following identities will be utilized in filter derivations: 
\begin{align}
\left[Ra\right]_{\times}= & R\left[a\right]_{\times}R^{\top},\hspace{1em}a\in{\rm \mathbb{R}}^{3},R\in\mathbb{SO}\left(3\right)\label{eq:SLAM_Identity1}\\
\left[b\times a\right]_{\times}= & ab^{\top}-ba^{\top},\hspace{1em}a,b\in{\rm \mathbb{R}}^{3}\label{eq:SLAM_Identity2}\\
\left[a\right]_{\times}^{2}= & -||a||^{2}\mathbf{I}_{3}+aa^{\top},\hspace{1em}a\in{\rm \mathbb{R}}^{3}\label{eq:SLAM_Identity4}\\
M\left[a\right]_{\times}+\left[a\right]_{\times}M= & {\rm Tr}\left\{ M\right\} \left[a\right]_{\times}-\left[Ma\right]_{\times},\nonumber \\
& \hspace{4em}a\in{\rm \mathbb{R}}^{3},M\in\mathbb{R}^{3\times3}\label{eq:SLAM_Identity5}\\
{\rm Tr}\left\{ \left[a\right]_{\times}M\right\} = & 0,\hspace{1em}a\in{\rm \mathbb{R}}^{3},M=M^{\top}\in\mathbb{R}^{3\times3}\label{eq:SLAM_Identity3}
\end{align}
\begin{align}
{\rm Tr}\left\{ M\left[a\right]_{\times}\right\} = & {\rm Tr}\left\{ \boldsymbol{\mathcal{P}}_{a}\left(M\right)\left[a\right]_{\times}\right\} =-2\mathbf{vex}\left(\boldsymbol{\mathcal{P}}_{a}\left(M\right)\right)^{\top}a,\nonumber \\
& \hspace{4em}a\in{\rm \mathbb{R}}^{3},M\in\mathbb{R}^{3\times3}\label{eq:SLAM_Identity6}
\end{align}

\section{SLAM Formulation in Stochastic Sense\label{sec:SE3_Problem-Formulation}}

The rigid-body's (vehicle's) attitude $R\in\mathbb{SO}\left(3\right)$,
a vital part of the robot's pose $\boldsymbol{T}\in\mathbb{SE}\left(3\right)$,
is expressed in the body-frame $R\in\left\{ \mathcal{B}\right\} $,
while its translation $P\in\mathbb{R}^{3}$ is expressed in the inertial-frame
$P\in\left\{ \mathcal{I}\right\} $. Let the map include $n$ landmarks
where ${\rm p}_{i}$ denotes location of the $i$th landmark defined
in the inertial-frame ${\rm p}_{i}\in\left\{ \mathcal{I}\right\} $
for all $i=1,2,\ldots,n$. SLAM problem considers the following two
elements to be completely unknown: 1) pose of the moving robot, and
2) landmarks within the environment $\overline{{\rm p}}=\left[\overline{{\rm p}}_{1},\overline{{\rm p}}_{2},\ldots,\overline{{\rm p}}_{n}\right]\in\overline{\mathcal{M}}^{n}$.
Accordingly, SLAM estimation problem given a set of measurements incorporates
two tasks executed concurrently: 1) estimation of the robot's pose
with respect to the environment landmarks, and 2) estimation of landmark
positions within the map. Figure \ref{fig:SLAM} illustrates the SLAM
estimation problem. 
\begin{figure*}
	\centering{}\includegraphics[scale=0.6]{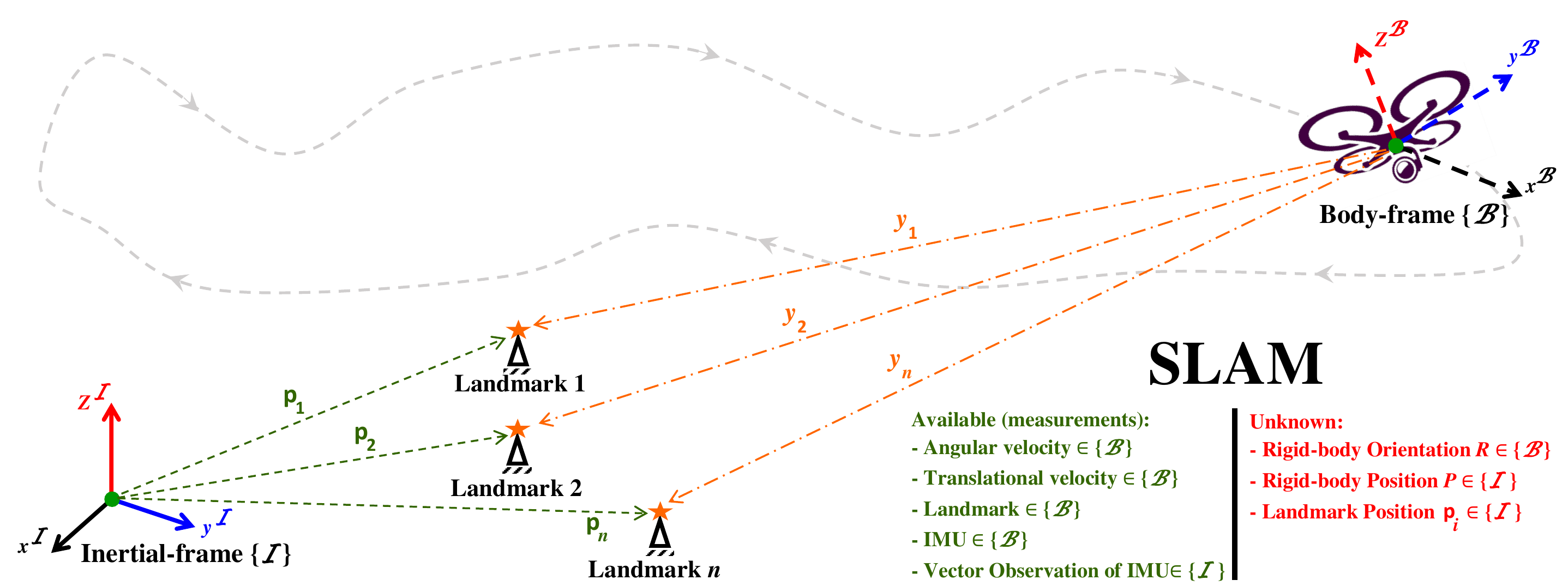}\caption{SLAM estimation problem.}
	\label{fig:SLAM}
\end{figure*}

\subsection{SLAM Kinematics and Measurements}

Let $X=\left(\boldsymbol{T},\overline{{\rm p}}\right)\in\mathbb{SLAM}_{n}\left(3\right)$
denote the true configuration of the SLAM problem with $\boldsymbol{T}\in\mathbb{SE}\left(3\right)$
as in \eqref{eq:T_SLAM} and $\overline{{\rm p}}=\left[\overline{{\rm p}}_{1},\overline{{\rm p}}_{2},\ldots,\overline{{\rm p}}_{n}\right]\in\overline{\mathcal{M}}^{n}$.
Notice that $X$ is unknown. A group of measurements is available
in $\left\{ \mathcal{B}\right\} $ and can be employed for SLAM estimation,
namely 1) body-frame measurements associated with attitude determination,
2) landmark measurements, and 3) group velocity measurements. Assume
that there are $n_{R}$ body-frame vectors suitable for attitude determination
and available for measurement defined by \cite{hashim2018SO3Stochastic,hashim2019SO3Wiley}
\[
\overset{\circ}{a}_{j}=\boldsymbol{T}^{-1}\overset{\circ}{r}_{j}+\overset{\circ}{b}_{j}^{a}+\overset{\circ}{n}_{j}^{a}\in\overset{\circ}{\mathcal{M}},\hspace{1em}j=1,2,\ldots,n_{R}
\]
or equivalently
\begin{equation}
a_{j}=R^{\top}r_{j}+b_{j}^{a}+n_{j}^{a}\in\mathbb{R}^{3}\label{eq:SLAM_Vect_R}
\end{equation}
where $r_{j}$ denotes known inertial-frame vector, $b_{j}^{a}$ denotes
unknown constant bias, and $n_{j}^{a}$ stands for unknown random
noise of the $j$th measurement. Note that the inverse of $\boldsymbol{T}$
is $\boldsymbol{T}^{-1}=\left[\begin{array}{cc}
R^{\top} & -R^{\top}P\\
\underline{\mathbf{0}}_{3}^{\top} & 1
\end{array}\right]\in\mathbb{SE}\left(3\right)$. The measurements in \eqref{eq:SLAM_Vect_R} exemplify a low cost
IMU. It is a common practice to normalize $r_{j}$ and $a_{j}$ in
\eqref{eq:SLAM_Vect_R} as follows
\begin{equation}
\upsilon_{j}^{r}=\frac{r_{j}}{\left\Vert r_{j}\right\Vert },\hspace{1em}\upsilon_{j}^{a}=\frac{a_{j}}{\left\Vert a_{j}\right\Vert }\label{eq:SLAM_Vector_norm}
\end{equation}
The normalized values in \eqref{eq:SLAM_Vector_norm} will be part
of the subsequent estimation. Consider combining the normalized vectors
into two distinct sets as follows
\begin{equation}
\begin{cases}
\upsilon^{r} & =\left[\upsilon_{1}^{r},\upsilon_{2}^{r},\ldots,\upsilon_{n_{R}}^{r}\right]\in\left\{ \mathcal{I}\right\} \\
\upsilon^{a} & =\left[\upsilon_{1}^{a},\upsilon_{2}^{a},\ldots,\upsilon_{n_{R}}^{a}\right]\in\left\{ \mathcal{B}\right\} 
\end{cases}\label{eq:SE3STCH_Set_R_Norm}
\end{equation}

\begin{rem}
	\label{rem:R_Marix}Rigid-body's attitude can be established provided
	that at least three non-collinear vectors in $\left\{ \mathcal{B}\right\} $
	along with their observations in $\left\{ \mathcal{I}\right\} $ are
	obtainable at each time sample. In case of $n_{R}=2$, the third measurement
	in $\left\{ \mathcal{B}\right\} $ and its observation in $\left\{ \mathcal{I}\right\} $
	is to be calculated via the cross product $\upsilon_{3}^{a}=\upsilon_{1}^{a}\times\upsilon_{2}^{a}$
	and $\upsilon_{3}^{r}=\upsilon_{1}^{r}\times\upsilon_{2}^{r}$, respectively
	ensuring that the two sets in \eqref{eq:SLAM_Vector_norm} are with
	rank 3.
\end{rem}
Assume that $n$ landmarks are available for measurement in the body-frame
via, for example, low-cost inertial vision units. The $i$th measurement
is as follows \cite{hashim2020SE3Stochastic,hashim2019SE3Det}:
\[
\overline{y}_{i}=\boldsymbol{T}^{-1}\overline{{\rm p}}_{i}+\overset{\circ}{b}_{i}^{y}+\overset{\circ}{n}_{i}^{y}\in\overline{\mathcal{M}},\hspace{1em}\forall i=1,2,\ldots,n
\]
or equivalently
\begin{equation}
y_{i}=R^{\top}\left({\rm p}_{i}-P\right)+b_{i}^{y}+n_{i}^{y}\in\mathbb{R}^{3}\label{eq:SLAM_Vec_Landmark}
\end{equation}
with $R$, $P$, and ${\rm p}_{i}$ representing the true attitude
and position of the robot, and landmark position, respectively, while
$b_{i}^{y}$ and $n_{i}^{y}$ stand for unknown constant bias and
random noise, respectively, for all $y_{i},b_{i}^{y},n_{i}^{y}\in\left\{ \mathcal{B}\right\} $.

\begin{assum}\label{Assumption:Feature}A minimum of three landmarks
	available for measurement is necessary to define a plane $\overline{y}=\left[\overline{y}_{1},\overline{y}_{2},\ldots,\overline{y}_{n}\right]\in\overline{\mathcal{M}}^{n}$.\end{assum}

Consider $\mathcal{Y}=\left(\left[U\right]_{\wedge},\overset{\circ}{{\rm v}}\right)\in\mathfrak{slam}_{n}\left(3\right)$
to be the true group velocity which is bounded and continuous with
$\overset{\circ}{{\rm v}}=\left[\overset{\circ}{{\rm v}}_{1},\overset{\circ}{{\rm v}}_{2},\ldots,\overset{\circ}{{\rm v}}_{n}\right]\in\overset{\circ}{\mathcal{M}}^{n}$.
Note that $\mathcal{Y}$ is given through sensor measurements. Hence,
the true motion dynamics of the vehicle's pose and $n$-landmarks
are
\begin{equation}
\begin{cases}
\dot{\boldsymbol{T}} & =\boldsymbol{T}\left[U\right]_{\wedge}\\
\dot{{\rm p}}_{i} & =R{\rm v}_{i},\hspace{1em}\forall i=1,2,\ldots,n
\end{cases}\label{eq:SLAM_True_dot}
\end{equation}
The dynamics in \eqref{eq:SLAM_True_dot} can be expressed as
\[
\begin{cases}
\dot{R} & =R\left[\Omega\right]_{\times}\\
\dot{P} & =RV\\
\dot{{\rm p}}_{i} & =R{\rm v}_{i},\hspace{1em}\forall i=1,2,\ldots,n
\end{cases}
\]
with $U=\left[\Omega^{\top},V^{\top}\right]^{\top}\in\mathbb{R}^{6}$
referring to the group velocity vector of the rigid-body where $\Omega\in\mathbb{R}^{3}$
is the true angular velocity and $V\in\mathbb{R}^{3}$ is the true
translational velocity. ${\rm v}_{i}\in\mathbb{R}^{3}$ defines the
$i$th linear velocity of the landmark in the moving-frame for all
$\Omega,V,{\rm v}_{i}\in\left\{ \mathcal{B}\right\} $. The measurements
of angular and translational velocity are defined as
\begin{equation}
\begin{cases}
\Omega_{m} & =\Omega+b_{\Omega}+n_{\Omega}\in\mathbb{R}^{3}\\
V_{m} & =V+b_{V}+n_{V}\in\mathbb{R}^{3}
\end{cases}\label{eq:SLAM_TVelcoity}
\end{equation}
where $b_{\Omega}$ and $b_{V}$ denote unknown constant bias, and
$n_{\Omega}$ and $n_{V}$ denote unknown random noise. Define the
group of velocity measurements, bias, and noise as $U_{m}=\left[\Omega_{m}^{\top},V_{m}^{\top}\right]^{\top}$,
$b_{U}=\left[b_{\Omega}^{\top},b_{V}^{\top}\right]^{\top}$, and $n_{U}=\left[n_{\Omega}^{\top},n_{V}^{\top}\right]^{\top}$,
respectively, for all $U_{m},b_{U},n_{U}\in\mathbb{R}^{6}$. This
work concerns exclusively fixed landmark environments, thereby $\dot{{\rm p}}_{i}=\underline{\mathbf{0}}_{3}$
and ${\rm v}_{i}=\underline{\mathbf{0}}_{3}$ $\forall i=1,2,\ldots,n$.

\subsection{SLAM Kinematics in Stochastic Sense}

\noindent Recall the expression of group velocity measurements in
\eqref{eq:SLAM_TVelcoity}. Since derivative of a Gaussian process
results a Gaussian process, the SLAM dynamics in \eqref{eq:SLAM_True_dot}
can be rewritten with respect to Brownian motion process vector $d\beta_{U}/dt\in\mathbb{R}^{6}$
\cite{khasminskii1980stochastic,jazwinski2007stochastic}. Assume
$\left\{ n_{U},t\geq t_{0}\right\} $ to be a vector representation
of the independent Brownian motion process 
\begin{equation}
n_{U}=\mathcal{Q}_{U}\frac{d\beta_{U}}{dt}\in\mathbb{R}^{6}\label{eq:SLAM_noise}
\end{equation}
where $\mathcal{Q}_{U}\in\mathbb{R}^{6\times6}$ denotes an unknown
nonzero nonnegative time-variant diagonal matrix whose elements are
bounded. The related covariance of the noise $n_{U}$ can be expressed
as $\mathcal{Q}_{U}^{2}=\mathcal{Q}_{U}\mathcal{Q}_{U}^{\top}$. The
following properties characterize the Brownian motion process \cite{ito1984lectures,jazwinski2007stochastic,deng2001stabilization,hashim2020SE3Stochastic,tong2011observer}:
\[
\mathbb{P}\left\{ \beta_{U}\left(0\right)=0\right\} =1,\hspace{1em}\mathbb{E}\left[d\beta_{U}/dt\right]=0,\hspace{1em}\mathbb{E}\left[\beta_{U}\right]=0
\]
In view of \eqref{eq:SLAM_True_dot}, \eqref{eq:SLAM_TVelcoity},
and \eqref{eq:SLAM_noise}, SLAM dynamics could be represented by
a stochastic differential equation 
\begin{equation}
\begin{cases}
d\boldsymbol{T} & =\boldsymbol{T}\left[U_{m}-b_{U}\right]_{\wedge}dt-\boldsymbol{T}\left[\mathcal{Q}_{U}d\beta_{U}\right]_{\wedge}\\
d{\rm p}_{i} & =R{\rm v}_{i}dt,\hspace{1em}\forall i=1,2,\ldots,n
\end{cases}\label{eq:SLAM_True_dot_STOCH}
\end{equation}
Or equivalently
\[
\begin{cases}
dR & =R\left[\Omega_{m}-b_{\Omega}\right]_{\times}dt-R\left[\mathcal{Q}_{\Omega}d\beta_{\Omega}\right]_{\times}\\
dP & =R\left(V_{m}-b_{V}\right)dt-R\mathcal{Q}_{V}d\beta_{V}\\
d{\rm p}_{i} & =R{\rm v}_{i}dt,\hspace{1em}\forall i=1,2,\ldots,n
\end{cases}
\]
where $U=U_{m}-b_{U}-n_{U}$ is considered. Given unknown bias $b_{U}$
and unknown time-variant covariance matrix $\mathcal{Q}_{U}$, with
the aim of achieving adaptive stabilization, define $\sigma$ as the
upper bound of $\mathcal{Q}_{U}^{2}$
\begin{equation}
\sigma=\left[\begin{array}{c}
{\rm max}\left\{ \mathcal{Q}_{\Omega\left(1,1\right)}^{2},\mathcal{Q}_{V\left(1,1\right)}^{2}\right\} \\
{\rm max}\left\{ \mathcal{Q}_{\Omega\left(2,2\right)}^{2},\mathcal{Q}_{V\left(2,2\right)}^{2}\right\} \\
{\rm max}\left\{ \mathcal{Q}_{\Omega\left(3,3\right)}^{2},\mathcal{Q}_{V\left(3,3\right)}^{2}\right\} 
\end{array}\right]\in\mathbb{R}^{3}\label{eq:SLAM_s_covariance}
\end{equation}
with ${\rm max}\left\{ \cdot\right\} $ being maximum value of the
corresponding elements. 

\begin{assum}\label{Assum:Boundedness} (Uniform boundedness of $b_{U}$
	and $\sigma$) Consider $b_{U}$ and $\sigma$ to belong to a known
	compact set $\varLambda_{U}$ with $b_{U},\sigma\in\varLambda_{U}\subset\mathbb{R}^{3}$,
	such that $b_{U}$ and $\sigma$ are upper bounded by a constant $\varPi$
	where $||\varLambda_{U}||\leq\overline{\varLambda}<\infty$.\end{assum}

\subsection{Error Criteria}

Define the pose estimate as
\[
\hat{\boldsymbol{T}}=\left[\begin{array}{cc}
\hat{R} & \hat{P}\\
\underline{\mathbf{0}}_{3}^{\top} & 1
\end{array}\right]\in\mathbb{SE}\left(3\right)
\]
with $\hat{R}$ being the estimate of $R$, and $\hat{P}$ being the
estimate of $P$ in \eqref{eq:T_SLAM}. Define $\overline{\hat{y}}_{i}=\hat{\boldsymbol{T}}^{-1}\overline{\hat{{\rm p}}}_{i}$
where $\hat{{\rm p}}_{i}$ is the $i$th landmark estimate of ${\rm p}_{i}$.
Let the pose error (true relative to estimated) be
\begin{align}
\tilde{\boldsymbol{T}}=\hat{\boldsymbol{T}}\boldsymbol{T}^{-1} & =\left[\begin{array}{cc}
\hat{R} & \hat{P}\\
\underline{\mathbf{0}}_{3}^{\top} & 1
\end{array}\right]\left[\begin{array}{cc}
R^{\top} & -R^{\top}P\\
\underline{\mathbf{0}}_{3}^{\top} & 1
\end{array}\right]\nonumber \\
& =\left[\begin{array}{cc}
\tilde{R} & \tilde{P}\\
\underline{\mathbf{0}}_{3}^{\top} & 1
\end{array}\right]\label{eq:SLAM_T_error}
\end{align}
with $\tilde{R}=\hat{R}R^{\top}$ being the error in orientation and
$\tilde{P}=\hat{P}-\tilde{R}P$ being the error in position of the
rigid-body. Pose estimation aims to asymptotically drive $\tilde{\boldsymbol{T}}\rightarrow\mathbf{I}_{4}$
in order to achieve this goal $\tilde{R}\rightarrow\mathbf{I}_{3}$
and $\tilde{P}\rightarrow\underline{\mathbf{0}}_{3}$. Let the landmark
position error (true relative to estimated) be
\begin{equation}
\overset{\circ}{e}_{i}=\overline{\hat{{\rm p}}}_{i}-\tilde{\boldsymbol{T}}\,\overline{{\rm p}}_{i},\hspace{1em}\forall i=1,2,\ldots,n\label{eq:SLAM_e}
\end{equation}
such that $\overset{\circ}{e}_{i}=\left[e_{i}^{\top},0\right]^{\top}\in\overset{\circ}{\mathcal{M}}$
and $\overline{\hat{{\rm p}}}_{i}=\left[\hat{{\rm p}}_{i}^{\top},1\right]^{\top}\in\overline{\mathcal{M}}$.
Note that $\overset{\circ}{e}_{i}=\overline{\hat{{\rm p}}}_{i}-\hat{\boldsymbol{T}}\boldsymbol{T}^{-1}\,\overline{{\rm p}}_{i}$,
and therefore from \eqref{eq:SLAM_Vec_Landmark} one has
\begin{equation}
\overset{\circ}{e}_{i}=\overline{\hat{{\rm p}}}_{i}-\hat{\boldsymbol{T}}\,\overline{y}_{i},\hspace{1em}\forall i=1,2,\ldots,n\label{eq:SLAM_e_Final}
\end{equation}
which leads to
\begin{align}
\overset{\circ}{e}_{i} & =\left[\begin{array}{c}
\hat{{\rm p}}_{i}\\
1
\end{array}\right]-\left[\begin{array}{cc}
\hat{R} & \hat{P}\\
\underline{\mathbf{0}}_{3}^{\top} & 1
\end{array}\right]\left[\begin{array}{c}
R^{\top}\left({\rm p}_{i}-P\right)\\
1
\end{array}\right]\nonumber \\
& =\left[\begin{array}{c}
\tilde{{\rm p}}_{i}-\tilde{P}\\
0
\end{array}\right]\in\overset{\circ}{\mathcal{M}}\label{eq:SLAM_e_tilde}
\end{align}
with $\tilde{{\rm p}}_{i}=\hat{{\rm p}}_{i}-\tilde{R}{\rm p}_{i}$
and $\tilde{P}=\hat{P}-\tilde{R}P$. Considering the fact that the
last row of the matrix in \eqref{eq:SLAM_e_tilde} is a zero, define
the stochastic differential equation of the error above as
\begin{align}
de_{i}= & \mathcal{F}_{i}dt+\mathcal{G}_{i}\mathcal{Q}_{U}d\beta_{U},\hspace{1em}\forall i=1,2,\ldots,n\label{eq:SLAM_e_dot_Stochastic}
\end{align}
where the stochastic dynamics in \eqref{eq:SLAM_e_dot_Stochastic}
are to be obtained in the stochastic filter. Taking in consideration
the group velocity in \eqref{eq:SLAM_TVelcoity} with $\hat{b}_{U}=\left[\hat{b}_{\Omega}^{\top},\hat{b}_{V}^{\top}\right]^{\top}$
being the unknown bias estimate of $b_{U}$, define the bias error
as
\begin{equation}
\begin{cases}
\tilde{b}_{\Omega} & =b_{\Omega}-\hat{b}_{\Omega}\\
\tilde{b}_{V} & =b_{V}-\hat{b}_{V}
\end{cases}\label{eq:SLAM_b_error}
\end{equation}
where $\tilde{b}_{U}=b_{U}-\hat{b}_{U}=\left[\tilde{b}_{\Omega}^{\top},\tilde{b}_{V}^{\top}\right]^{\top}\in\mathbb{R}^{6}$.
Also, consider $\hat{\sigma}$ to be the estimate of $\sigma$ in
\eqref{eq:SLAM_s_covariance}. Define the error between $\hat{\sigma}$
and $\sigma$ as follows
\begin{equation}
\tilde{\sigma}=\sigma-\hat{\sigma}\label{eq:SLAM_s_error}
\end{equation}
The subsequent Definitions and Lemmas are applicable in the derivation
process of the nonlinear stochastic estimator for SLAM.
\begin{defn}
	\label{def:Unstable-set}Let $\mathcal{U}_{s}$ be a non-attractive
	forward invariant unstable subset of $\mathbb{SO}\left(3\right)$
	\begin{equation}
	\mathcal{U}_{s}=\left\{ \left.\tilde{R}\left(0\right)\in\mathbb{SO}\left(3\right)\right|{\rm Tr}\{\tilde{R}(0)\}=-1\right\} \label{eq:SO3_PPF_STCH_SET}
	\end{equation}
	The only three possible scenarios for $\tilde{R}\left(0\right)\in\mathcal{U}_{s}$
	are: $\tilde{R}\left(0\right)={\rm diag}(1,-1,-1)$, $\tilde{R}\left(0\right)={\rm diag}(-1,1,-1)$,
	and $\tilde{R}\left(0\right)={\rm diag}(-1,-1,1)$.
\end{defn}
\begin{lem}
	\label{Lemm:SLAM_RM_I2}Define $\tilde{R}\in\mathbb{SO}\left(3\right)$,
	$M=M^{\top}\in\mathbb{R}^{3\times3}$ such that ${\rm rank}\left\{ M\right\} =3$
	and ${\rm Tr}\left\{ M\right\} =3$. Define $\breve{\mathbf{M}}={\rm Tr}\left\{ M\right\} \mathbf{I}_{3}-M$
	with $\underline{\lambda}=\underline{\lambda}(\breve{\mathbf{M}})$
	being the minimum singular value of $\breve{\mathbf{M}}$. Thereby,
	the following holds:
	\begin{align}
	||\tilde{R}M||_{{\rm I}} & \leq\frac{2}{\underline{\lambda}}\frac{||\mathbf{vex}\left(\boldsymbol{\mathcal{P}}_{a}(\tilde{R}M)\right)||^{2}}{1+{\rm Tr}\{\tilde{R}MM^{-1}\}}\label{eq:SLAM_lemm1_2}
	\end{align}
	\textbf{Proof. See Lemma 1 \cite{hashim2019SO3Wiley}.} 
\end{lem}
\begin{defn}
	\label{def:SLAM_LV} Consider the stochastic differential system in
	\eqref{eq:SLAM_e_dot_Stochastic}, and let $\boldsymbol{{\rm V}}\left(e_{1},\ldots,e_{n}\right)$
	be a twice differentiable function $\boldsymbol{{\rm V}}\left(e_{1},\ldots,e_{n}\right)\in\mathcal{C}^{2}$.
	The differential operator $\mathcal{L}\boldsymbol{{\rm V}}\left(e_{1},\ldots,e_{n}\right)$
	is expressed as below 
	\[
	\mathcal{L}\boldsymbol{{\rm V}}\left(e_{1},\ldots,e_{n}\right)=\sum_{i=1}^{n}\left(\boldsymbol{{\rm V}}_{e_{i}}^{\top}\mathcal{F}_{i}+\frac{1}{2}{\rm Tr}\left\{ \mathcal{G}_{i}\mathcal{Q}_{U}^{2}\mathcal{G}_{i}^{\top}\boldsymbol{{\rm V}}_{e_{i}e_{i}}\right\} \right)
	\]
	such that $\boldsymbol{{\rm V}}_{e_{i}}=\partial\boldsymbol{{\rm V}}/\partial e_{i}$,
	and $\boldsymbol{{\rm V}}_{e_{i}e_{i}}=\partial^{2}\boldsymbol{{\rm V}}/\partial e_{i}^{2}$
	$\forall i=1,2,\ldots,n$.
\end{defn}
\begin{defn}
	\label{def:SLAM_SGUUB}\cite{hashim2018SO3Stochastic,hashim2020SE3Stochastic,ji2006adaptive}
	Consider the stochastic differential system in \eqref{eq:SLAM_e_dot_Stochastic}
	with trajectory $e_{i}$ being SGUUB if for a given compact set $\Sigma\in\mathbb{R}^{4}$
	and any $e_{i}\left(t_{0}\right)$, there exists a positive constant
	$\kappa>0$, and a time constant $\tau=\tau\left(\kappa,e_{i}\left(t_{0}\right)\right)$
	with $\mathbb{E}\left[\left\Vert e_{i}\left(t_{0}\right)\right\Vert \right]<\kappa,\forall t>t_{0}+\tau$. 
\end{defn}
\begin{lem}
	\label{Lemm:SLAM_deng} \cite{deng2001stabilization} Consider the
	stochastic dynamics in \eqref{eq:SLAM_e_dot_Stochastic} to be assigned
	with a potential function $\boldsymbol{{\rm V}}\in\mathcal{C}^{2}$
	where $\boldsymbol{{\rm V}}:\mathbb{R}^{3}\rightarrow\mathbb{R}_{+}$.
	Suppose there exist class $\mathcal{K}_{\infty}$ functions $\bar{\alpha}_{1}\left(\cdot\right)$
	and $\bar{\alpha}_{2}\left(\cdot\right)$, constants $\eta_{1}>0$
	and $\eta_{2}\geq0$ such that 
	\begin{equation}
	\bar{\alpha}_{1}\left(e_{1},\ldots,e_{n}\right)\leq\boldsymbol{{\rm V}}\left(e_{1},\ldots,e_{n}\right)\leq\bar{\alpha}_{2}\left(e_{1},\ldots,e_{n}\right)\label{eq:SLAM_Vfunction_Lyap}
	\end{equation}
	\begin{align}
	\mathcal{L}\boldsymbol{{\rm V}}\left(e_{1},\ldots,e_{n}\right)= & \sum_{i=1}^{n}\left(\boldsymbol{{\rm V}}_{e_{i}}^{\top}\mathcal{F}_{i}+\frac{1}{2}{\rm Tr}\left\{ \mathcal{G}_{i}\mathcal{Q}_{U}^{2}\mathcal{G}_{i}^{\top}\boldsymbol{{\rm V}}_{e_{i}e_{i}}\right\} \right)\nonumber \\
	\leq & -\eta_{1}\boldsymbol{{\rm V}}\left(e_{1},\ldots,e_{n}\right)+\eta_{2}\label{eq:SLAM_dVfunction_Lyap}
	\end{align}
	then for $e_{i}\in\mathbb{R}^{4}$, there is almost a unique strong
	solution on $\left[0,\infty\right)$ for the stochastic dynamics in
	\eqref{eq:SLAM_e_dot_Stochastic}. Moreover, the solution $e_{i}$
	is bounded in probability where
	\begin{equation}
	\mathbb{E}\left[\boldsymbol{{\rm V}}\left(e_{1},\ldots,e_{n}\right)\right]\leq\boldsymbol{{\rm V}}\left(e_{1}\left(0\right),\ldots,e_{n}\left(0\right)\right){\rm exp}\left(-\eta_{1}t\right)+\frac{\eta_{2}}{\eta_{1}}\label{eq:SLAM_EVfunction_Lyap}
	\end{equation}
	In addition, if the inequality in \eqref{eq:SLAM_EVfunction_Lyap}
	is met, $e_{i}$ in \eqref{eq:SLAM_e_dot_Stochastic} is SGUUB in
	the mean square. 
\end{lem}
The existence of a unique solution and proof of Lemma \ref{Lemm:SLAM_deng}
can be found in \cite{deng2001stabilization}.
\begin{lem}
	\label{lem:SLAM_Young} (Young’s inequality) Let $a\in\mathbb{R}^{n}$
	and $b\in\mathbb{R}^{n}$. Define $c_{1}>1$ and $c_{2}>1$ such that
	$\left(c_{1}-1\right)\left(c_{2}-1\right)=1$, and $\varrho>0$ as
	a small constant. Consequently, the following holds:
	\begin{align}
	a^{\top}b & \leq\left(1/c_{1}\right)\varrho^{c_{1}}\left\Vert a\right\Vert ^{c_{1}}+\left(1/c_{2}\right)\varrho^{-c_{2}}\left\Vert b\right\Vert ^{c_{2}}\label{eq:SO3STCH_lem_ineq}
	\end{align}
\end{lem}
Prior to moving forward, it is important to recall that the true SLAM
dynamics in \eqref{eq:SLAM_True_dot} 1) are nonlinear and 2) are
posed on the Lie group of $\mathbb{SLAM}_{n}\left(3\right)=\mathbb{SE}\left(3\right)\times\overline{\mathcal{M}}^{n}$
where $X=\left(\boldsymbol{T},\overline{{\rm p}}\right)\in\mathbb{SLAM}_{n}\left(3\right)$.
Additionally, the tangent space of $\mathbb{SLAM}_{n}\left(3\right)$
is $\mathfrak{slam}_{n}\left(3\right)=\mathfrak{se}\left(3\right)\times\overset{\circ}{\mathcal{M}}^{n}$
such that and $\mathcal{Y}=\left(\left[U\right]_{\wedge},\overset{\circ}{{\rm v}}\right)\in\mathfrak{slam}_{n}\left(3\right)$.
With the aim of proposing a robust stochastic filter able to produce
good results, the proposed filter design should imitate the true nonlinearity
of the SLAM problem and should be modeled on the Lie group of $\mathbb{SLAM}_{n}\left(3\right)$
with the tangent space $\mathfrak{slam}_{n}\left(3\right)$. Complying
with the above-mentioned requirements, the structure of the stochastic
filter is $\hat{X}=\left(\hat{\boldsymbol{T}},\overline{\hat{{\rm p}}}\right)\in\mathbb{SLAM}_{n}\left(3\right)$
and $\hat{\mathcal{Y}}=\left([\hat{U}]_{\wedge},\overset{\circ}{\hat{{\rm v}}}\right)\in\mathfrak{slam}_{n}\left(3\right)$
with $\hat{\boldsymbol{T}}\in\mathbb{SE}\left(3\right)$ and $\overline{\hat{{\rm p}}}=\left[\overline{\hat{{\rm p}}}_{1},\ldots,\overline{\hat{{\rm p}}}_{n}\right]\in\overline{\mathcal{M}}^{n}$
being pose estimates and landmark positions, respectively, and $\hat{U}\in\mathfrak{se}\left(3\right)$
and $\overset{\circ}{\hat{{\rm v}}}=\left[\overset{\circ}{\hat{{\rm v}}}_{1},\ldots,\overset{\circ}{\hat{{\rm v}}}_{n}\right]\in\overset{\circ}{\mathcal{M}}^{n}$
being velocities to be designed in the following Section. It is worth
noting that $\overset{\circ}{\hat{{\rm v}}}_{i}=\left[\hat{{\rm v}}_{i}^{\top},0\right]\in\overset{\circ}{\mathcal{M}}$
and $\overline{\hat{{\rm p}}}_{i}=\left[\hat{{\rm p}}_{i}^{\top},1\right]^{\top}\in\overline{\mathcal{M}}$
for all $i=1,2,\ldots,n$ and $\hat{{\rm v}}_{i},\hat{{\rm p}}_{i}\in\mathbb{R}^{3}$.

\section{Nonlinear Stochastic Filter Design \label{sec:SLAM_Filter}}

The SLAM nonlinear stochastic filter design is proposed in this Section.
With the aim of defining the concept of the nonlinear SLAM filtering
and paving the way for the novel nonlinear stochastic filter solution
presented in the second subsection, the first subsection introduces
a nonlinear deterministic filter that operates based only on the surrounding
landmark measurements which is similar in the structure to \cite{hashim2020LetterSLAM,zlotnik2018SLAM}.
In contrast to the deterministic filter, the novel nonlinear stochastic
SLAM filter relies on measurements collected by a low-cost IMU and
measurements of the landmarks. The first simple filter will provide
a benchmark for the proposed stochastic solution. 

\subsection{Nonlinear Deterministic Filter Design without IMU\label{subsec:Det_without_IMU}}

Consider the nonlinear filter design for SLAM:

\begin{align}
\dot{\hat{\boldsymbol{T}}} & =\hat{\boldsymbol{T}}\left[U_{m}-\hat{b}_{U}-W_{U}\right]_{\wedge}\label{eq:SLAM_T_est_dot_f1}\\
\dot{{\rm \hat{p}}}_{i} & =-k_{p}e_{i},\hspace{1em}i=1,2,\ldots,n\label{eq:SLAM_p_est_dot_f1}\\
\dot{\hat{b}}_{U} & =-\sum_{i=1}^{n}\frac{\Gamma}{\alpha_{i}}\left[\begin{array}{c}
\left[y_{i}\right]_{\times}\hat{R}^{\top}\\
\hat{R}^{\top}
\end{array}\right]e_{i}\label{eq:SLAM_b_est_dot_f1}\\
W_{U} & =-\sum_{i=1}^{n}\frac{k_{w}}{\alpha_{i}}\left[\begin{array}{c}
\left[y_{i}\right]_{\times}\hat{R}^{\top}\\
\hat{R}^{\top}
\end{array}\right]e_{i}\label{eq:SLAM_W_f1}
\end{align}
with $k_{w}$, $k_{p}$, $\Gamma$, and $\alpha_{i}$ being positive
constants, $e_{i}$ being as given in \eqref{eq:SLAM_e_Final} for
all $i=1,2,\cdots,n$, $W_{U}=\left[W_{\Omega}^{\top},W_{V}^{\top}\right]^{\top}\in\mathbb{R}^{6}$
being a correction factor, and $\hat{b}_{U}=\left[\hat{b}_{\Omega}^{\top},\hat{b}_{V}^{\top}\right]^{\top}\in\mathbb{R}^{6}$
being the estimate of $b_{U}$.
\begin{thm}
	Consider the true motion of SLAM dynamics to be $\dot{X}=\left(\dot{\boldsymbol{T}},\dot{\overline{{\rm p}}}\right)$
	as in \eqref{eq:SLAM_True_dot}, the output to be landmark measurements
	($\overline{y}_{i}=\boldsymbol{T}^{-1}\overline{{\rm p}}_{i}$) for
	all $i=1,2,\ldots,n$ and the velocity measurements in \eqref{eq:SLAM_True_dot}
	to be attached only with constant bias where $U_{m}=U+b_{U}$ and
	$n_{U}=0$. Let Assumption \ref{Assumption:Feature} hold true and
	the deterministic filter be as in \eqref{eq:SLAM_T_est_dot_f1}, \eqref{eq:SLAM_p_est_dot_f1},
	\eqref{eq:SLAM_b_est_dot_f1}, and \eqref{eq:SLAM_W_f1} combined
	with the measurements of $U_{m}$ and $\overline{y}_{i}$. Set the
	design parameters $k_{w}$, $k_{p}$, $\Gamma$, and $\alpha_{i}$
	as positive scalars for all $i=1,2,\ldots,n$. Also, consider the
	set
	\begin{align}
	\mathcal{S}= & \left\{ \left(e_{1},e_{2},\ldots,e_{n}\right)\in\mathbb{R}^{3}\times\mathbb{R}^{3}\times\cdots\times\mathbb{R}^{3}\right|\nonumber \\
	& \hspace{9em}e_{i}=\underline{\mathbf{0}}_{3}\forall i=1,2,\ldots n\}\label{eq:SLAM_Set1}
	\end{align}
	Then 1) the error $e_{i}$ in \eqref{eq:SLAM_e} is exponentially
	regulated to the set $\mathcal{S}$, 2) $\tilde{\boldsymbol{T}}$
	remains bounded and 3) given constants $R_{c}\in\mathbb{SO}\left(3\right)$
	and $P_{c}\in\mathbb{R}^{3}$ one has $\tilde{R}\rightarrow R_{c}$
	and $\tilde{P}\rightarrow P_{c}$ as $t\rightarrow\infty$. 
\end{thm}
\begin{proof}Since $\boldsymbol{\dot{T}}^{-1}=-\boldsymbol{T}^{-1}\boldsymbol{\dot{T}}\boldsymbol{T}^{-1}$,
	one obtains the error dynamics of $\tilde{\boldsymbol{T}}$ defined
	in \eqref{eq:SLAM_T_error} as follows
	\begin{align}
	\dot{\tilde{\boldsymbol{T}}} & =\dot{\hat{\boldsymbol{T}}}\boldsymbol{T}^{-1}+\hat{\boldsymbol{T}}\dot{\boldsymbol{T}}^{-1}\nonumber \\
	& =\hat{\boldsymbol{T}}\left[U+\tilde{b}_{U}-W_{U}\right]_{\wedge}\boldsymbol{T}^{-1}-\hat{\boldsymbol{T}}\left[U\right]_{\wedge}\boldsymbol{T}^{-1}\nonumber \\
	& =\hat{\boldsymbol{T}}\left[\tilde{b}_{U}-W_{U}\right]_{\wedge}\hat{\boldsymbol{T}}^{-1}\tilde{\boldsymbol{T}}\label{eq:SLAM_T_error_dot}
	\end{align}
	Thereby, the error dynamics of $\overset{\circ}{e}_{i}$ in \eqref{eq:SLAM_e}
	are
	\begin{align}
	\overset{\circ}{\dot{e}}_{i} & =\overset{\circ}{\dot{\hat{{\rm p}}}}_{i}-\dot{\tilde{\boldsymbol{T}}}\,\overline{{\rm p}}_{i}-\tilde{\boldsymbol{T}}\,\dot{\overline{{\rm p}}}_{i}\nonumber \\
	& =\overset{\circ}{\dot{\hat{{\rm p}}}}_{i}-\hat{\boldsymbol{T}}\left[\tilde{b}_{U}-W_{U}\right]_{\wedge}\hat{\boldsymbol{T}}^{-1}\tilde{\boldsymbol{T}}\,\overline{{\rm p}}_{i}\label{eq:SLAM_e_dot}
	\end{align}
	From \eqref{eq:SLAM_T_error_dot}, one finds
	\begin{align}
	\hat{\boldsymbol{T}}\left[\tilde{b}_{U}\right]_{\wedge}\hat{\boldsymbol{T}}^{-1} & =\left[\begin{array}{cc}
	\hat{R} & \hat{P}\\
	\underline{\mathbf{0}}_{3}^{\top} & 1
	\end{array}\right]\left[\begin{array}{cc}
	\left[\tilde{b}_{\Omega}\right]_{\times} & \tilde{b}_{V}\\
	\underline{\mathbf{0}} & 0
	\end{array}\right]\left[\begin{array}{cc}
	\hat{R}^{\top} & -\hat{R}^{\top}\hat{P}\\
	\underline{\mathbf{0}}_{3}^{\top} & 1
	\end{array}\right]\nonumber \\
	& =\left[\begin{array}{cc}
	\hat{R}\left[\tilde{b}_{\Omega}\right]_{\times}\hat{R}^{\top} & \hat{R}\tilde{b}_{V}-\hat{R}\left[\tilde{b}_{\Omega}\right]_{\times}\hat{R}^{\top}\hat{P}\\
	\underline{\mathbf{0}}_{3}^{\top} & 0
	\end{array}\right]\nonumber \\
	& =\left[\begin{array}{c}
	\hat{R}\tilde{b}_{\Omega}\\
	\hat{R}\tilde{b}_{V}+\left[\hat{P}\right]_{\times}\hat{R}\tilde{b}_{\Omega}
	\end{array}\right]_{\wedge}\in\mathfrak{se}\left(3\right)\label{eq:SLAM_Adj_Property4}
	\end{align}
	where $\left[R\tilde{b}_{\Omega}\right]_{\times}=R\left[\tilde{b}_{\Omega}\right]_{\times}R^{\top}$
	as defined in \eqref{eq:SLAM_Identity1}. Recalling the definition
	of wedge operator in \eqref{eq:SLAM_wedge}, one finds that \eqref{eq:SLAM_Adj_Property4}
	becomes
	\begin{equation}
	\hat{\boldsymbol{T}}\left[\tilde{b}_{U}\right]_{\wedge}\hat{\boldsymbol{T}}^{-1}=\left[\left[\begin{array}{cc}
	\hat{R} & 0_{3\times3}\\
	\left[\hat{P}\right]_{\times}\hat{R} & \hat{R}
	\end{array}\right]\tilde{b}_{U}\right]_{\wedge}\label{eq:SLAM_Adj_Property5}
	\end{equation}
	According to \eqref{eq:SLAM_Adj_Property5} and \eqref{eq:SLAM_e_dot},
	one has
	\begin{align}
	\hat{\boldsymbol{T}}\left[\tilde{b}_{U}\right]_{\wedge}\hat{\boldsymbol{T}}^{-1}\tilde{\boldsymbol{T}}\,\overline{{\rm p}}_{i} & =\left[\left[\begin{array}{cc}
	\hat{R} & 0_{3\times3}\\
	\left[\hat{P}\right]_{\times}\hat{R} & \hat{R}
	\end{array}\right]\tilde{b}_{U}\right]_{\wedge}\left[\begin{array}{c}
	\hat{R}y_{i}+\hat{P}\\
	1
	\end{array}\right]\nonumber \\
	& =\left[\begin{array}{c}
	-\left[\hat{R}y_{i}\right]_{\times}\hat{R}\tilde{b}_{\Omega}+\hat{R}\tilde{b}_{V}\\
	0
	\end{array}\right]\nonumber \\
	& =\left[\begin{array}{cc}
	-\hat{R}\left[y_{i}\right]_{\times} & \hat{R}\\
	\underline{\mathbf{0}}_{3}^{\top} & \underline{\mathbf{0}}_{3}^{\top}
	\end{array}\right]\tilde{b}_{U}\label{eq:SLAM_Adj_Property6}
	\end{align}
	In view of \eqref{eq:SLAM_e_dot} and \eqref{eq:SLAM_Adj_Property6},
	one can rewrite \eqref{eq:SLAM_e_dot} as
	\begin{align}
	\overset{\circ}{\dot{e}}_{i} & =\overset{\circ}{\dot{\hat{{\rm p}}}}_{i}-\left[\begin{array}{cc}
	-\hat{R}\left[y_{i}\right]_{\times} & \hat{R}\\
	\underline{\mathbf{0}}_{3}^{\top} & \underline{\mathbf{0}}_{3}^{\top}
	\end{array}\right]\left(\tilde{b}_{U}-W_{U}\right)\label{eq:SLAM_e_dot1}
	\end{align}
	The last row in \eqref{eq:SLAM_e_dot1} are zeros, thereby, one obtains
	\begin{align}
	\dot{e}_{i} & =\dot{\hat{{\rm p}}}_{i}-\left[\begin{array}{cc}
	-\hat{R}\left[y_{i}\right]_{\times} & \hat{R}\end{array}\right]\left(\tilde{b}_{U}-W_{U}\right)\label{eq:SLAM_e_dot_Final}
	\end{align}
	Consider the candidate Lyapunov function $\boldsymbol{{\rm V}}=\boldsymbol{{\rm V}}\left(e_{1},\ldots,e_{n},\tilde{b}_{U}\right)$
	\begin{equation}
	\boldsymbol{{\rm V}}=\sum_{i=1}^{n}\frac{1}{2\alpha_{i}}e_{i}^{\top}e_{i}+\frac{1}{2}\tilde{b}_{U}^{\top}\Gamma^{-1}\tilde{b}_{U}\label{eq:SLAM_Lyap1}
	\end{equation}
	The time derivative of \eqref{eq:SLAM_Lyap1} is
	\begin{align}
	\dot{\boldsymbol{{\rm V}}}= & \sum_{i=1}^{n}\frac{1}{\alpha_{i}}e_{i}^{\top}\dot{e}_{i}-\tilde{b}_{U}^{\top}\Gamma^{-1}\dot{\hat{b}}_{U}\nonumber \\
	= & \sum_{i=1}^{n}\frac{1}{\alpha_{i}}e_{i}^{\top}\dot{\hat{{\rm p}}}_{i}-\sum_{i=1}^{n}\frac{1}{\alpha_{i}}e_{i}^{\top}\left[\begin{array}{cc}
	-\hat{R}\left[y_{i}\right]_{\times} & \hat{R}\end{array}\right]\left(\tilde{b}_{U}-W_{U}\right)\nonumber \\
	& -\tilde{b}_{U}^{\top}\Gamma^{-1}\dot{\hat{b}}_{U}\label{eq:SLAM_Lyap1_dot}
	\end{align}
	Replacing $W_{U}$, $\dot{\hat{b}}_{U}$ and $\dot{\hat{{\rm p}}}_{i}$
	with their expressions in \eqref{eq:SLAM_p_est_dot_f1}, \eqref{eq:SLAM_b_est_dot_f1}
	and \eqref{eq:SLAM_W_f1}, respectively, one obtains
	\begin{align}
	\dot{\boldsymbol{{\rm V}}}= & -\sum_{i=1}^{n}\frac{k_{p}}{\alpha_{i}}\left\Vert e_{i}\right\Vert ^{2}-k_{w}\left\Vert \sum_{i=1}^{n}\frac{e_{i}}{\alpha_{i}}\right\Vert ^{2}\nonumber \\
	& -k_{w}\left\Vert \sum_{i=1}^{n}\left[y_{i}\right]_{\times}\hat{R}^{\top}\frac{e_{i}}{\alpha_{i}}\right\Vert ^{2}\label{eq:SLAM_Lyap1_dot_Final}
	\end{align}
	Based on \eqref{eq:SLAM_Lyap1_dot_Final} the time derivative of $\boldsymbol{{\rm V}}$
	is negative definite where $\dot{\boldsymbol{{\rm V}}}$ equals to
	zero at $e_{i}=\underline{\mathbf{0}}_{3}$. The result in \eqref{eq:SLAM_Lyap1_dot_Final}
	affirms that $e_{i}$ is exponentially regulated to the set $\mathcal{S}$
	given in \eqref{eq:SLAM_Set1}. Based on Barbalat Lemma, $\dot{\boldsymbol{{\rm V}}}$
	is negative, continuous and approaches the origin implying that $\tilde{\boldsymbol{T}}$,
	$\tilde{b}_{U}$, and $\ddot{e}_{i}$ stay bounded. Also, the expression
	in \eqref{eq:SLAM_e_tilde} demonstrates that if $e_{i}\rightarrow\underline{\mathbf{0}}_{3}$,
	then $\tilde{{\rm p}}_{i}-\tilde{P}\rightarrow\underline{\mathbf{0}}_{3}$,
	and accordingly by \eqref{eq:SLAM_e} one has $\overline{\hat{{\rm p}}}_{i}-\tilde{\boldsymbol{T}}\,\overline{{\rm p}}_{i}\rightarrow\underline{\mathbf{0}}_{4}$.
	As such, $\tilde{\boldsymbol{T}}$ is upper bounded with $\tilde{R}\rightarrow R_{c}$
	and $\tilde{P}\rightarrow P_{c}$ as $t\rightarrow\infty$ which completes
	the proof.\end{proof}

\subsection{Nonlinear Stochastic Filter Design with IMU\label{subsec:Det_with_IMU}}

The nonlinear deterministic filter design in Subsection \ref{subsec:Det_without_IMU}
allows $e_{i}=\tilde{{\rm p}}_{i}-\tilde{P}\rightarrow\underline{\mathbf{0}}_{3}$
exponentially where $\tilde{P}=\hat{P}-\tilde{R}P$ and $\tilde{{\rm p}}_{i}=\hat{{\rm p}}_{i}-\tilde{R}{\rm p}_{i}$.
However, $\tilde{R}\rightarrow R_{c}$ and $\tilde{P}\rightarrow P_{c}$
as $t\rightarrow\infty$ such that $R_{c}\in\mathbb{SO}\left(3\right)$
and $P_{c}\in\mathbb{R}^{3}$ are constants. Hence, for the case when
$R\left(0\right)$ and $P\left(0\right)$ are not precisely known,
the error in $\tilde{R}$, $\tilde{P}$, and $\tilde{{\rm p}}_{i}$
will become remarkably large, resulting in highly inaccurate pose
and landmark estimates. Additionally, the nonlinear deterministic
filter considers the group velocity vector measurements associated
with the SLAM dynamics in \eqref{eq:SLAM_True_dot_STOCH} to be noise
free $n_{U}=0$. Failing to incorporate the impact of noise may significantly
undermine the effectiveness of the estimation process and destabilize
the overall closed loop dynamics. This behavior is exemplified by
the previously proposed solutions, for example \cite{hashim2020LetterSLAM,zlotnik2018SLAM}.
\begin{rem}
	\label{rem:SLAM-Observability} Define $R_{c}\in\mathbb{SO}\left(3\right)$
	and $P_{c}\in\mathbb{R}^{3}$ as constants. It has been definitively
	proven that SLAM problem is not observable \cite{lee2006SLAM_observability},
	therefore, the best achievable solution is for $\tilde{R}\rightarrow R_{c}$,
	$\tilde{P}\rightarrow P_{c}$, and $\hat{{\rm p}}_{i}\rightarrow\hat{P}+\tilde{R}{\rm p}_{i}-\tilde{R}P$
	as $t\rightarrow\infty$.
\end{rem}
Based on the above discussion, the objective of this subsection is
to propose a nonlinear stochastic filter design for SLAM that is able
to produce good performance through velocity, IMU, and landmark measurements
regardless the initial value of pose and landmarks were accurately
known or not. Considering the body-frame measurements and the associated
normalization in \eqref{eq:SLAM_Vect_R} and \eqref{eq:SLAM_Vector_norm},
define
\begin{equation}
M=M^{\top}=\sum_{j=1}^{n_{{\rm R}}}s_{j}\upsilon_{j}^{r}\left(\upsilon_{j}^{r}\right)^{\top},\hspace{1em}\forall j=1,2,\ldots n_{{\rm R}}\label{eq:SLAM_M}
\end{equation}
with $s_{j}\geq0$ being a constant gain associated with the confidence
level of the $j$th sensor measurements. Notice that $M$ in \eqref{eq:SLAM_M}
is symmetric. Based on Remark \ref{rem:R_Marix}, the availability
of a minimum two non-collinear body-frame measurements along with
their inertial-frame observations is assumed ($n_{{\rm R}}\geq2$)
which can be satisfied by a low-cost IMU module. For $n_{{\rm R}}=2$,
the third measurement and its observations are calculated using cross
product $\upsilon_{3}^{a}=\upsilon_{1}^{a}\times\upsilon_{2}^{a}$
and $\upsilon_{3}^{r}=\upsilon_{1}^{r}\times\upsilon_{2}^{r}$. As
such, ${\rm rank}\left(M\right)=3$. Defining the eigenvalues of $M$
as $\lambda\left(M\right)=\left\{ \lambda_{1},\lambda_{2},\lambda_{3}\right\} $,
one has $\lambda_{1},\lambda_{2},\lambda_{3}>0$. Let $\breve{\mathbf{M}}={\rm Tr}\left\{ M\right\} \mathbf{I}_{3}-M$,
given that ${\rm rank}\left(M\right)=3$. Hence, ${\rm rank}(\breve{\mathbf{M}})=3$
as well allowing to conclude that (\cite{bullo2004geometric} page.
553): 
\begin{enumerate}
	\item $\breve{\mathbf{M}}$ is positive-definite.
	\item $\breve{\mathbf{M}}$ has the following eigenvalues: $\lambda(\breve{\mathbf{M}})=\left\{ \lambda_{1}+\lambda_{2},\lambda_{2}+\lambda_{3},\lambda_{3}+\lambda_{1}\right\} $
	with $\underline{\lambda}(\breve{\mathbf{M}})>0$ being the minimum
	eigenvalue. 
\end{enumerate}
In all of the following discussions it is assumed that ${\rm rank}\left(M\right)=3$.
Additionally, for $j=1,2,\ldots,n_{{\rm R}}$ it is selected that
$\sum_{j=1}^{n_{{\rm R}}}s_{j}=3$ signifying that ${\rm Tr}\left\{ M\right\} =3$.

With the aim of proposing a stochastic filter design reliant on a
set of measurements, let us reintroduce the necessary variables in
vectorial terms. From \eqref{eq:SLAM_Vect_R} and \eqref{eq:SLAM_Vector_norm},
as the true normalized value of the $j$th body-frame vector is $\upsilon_{j}^{a}=R^{\top}\upsilon_{j}^{r}$,
let
\begin{equation}
\hat{\upsilon}_{j}^{a}=\hat{R}^{\top}\upsilon_{j}^{r},\hspace{1em}\forall j=1,2,\ldots n_{{\rm R}}\label{eq:SLAM_vect_R_estimate}
\end{equation}
Define the pose error analogously to \eqref{eq:SLAM_T_error} where
$\tilde{R}=\hat{R}R^{\top}$. Based on the identities in \eqref{eq:SLAM_Identity1}
and \eqref{eq:SLAM_Identity2}, one has
\begin{align*}
\left[\hat{R}\sum_{j=1}^{n_{{\rm R}}}\frac{s_{j}}{2}\hat{\upsilon}_{j}^{a}\times\upsilon_{j}^{a}\right]_{\times} & =\hat{R}\sum_{j=1}^{n_{{\rm R}}}\frac{s_{j}}{2}\left(\upsilon_{j}^{a}\left(\hat{\upsilon}_{j}^{a}\right)^{\top}-\hat{\upsilon}_{j}^{a}\left(\upsilon_{j}^{a}\right)^{\top}\right)\hat{R}^{\top}\\
& =\frac{1}{2}\hat{R}R^{\top}M-\frac{1}{2}MR\hat{R}^{\top}\\
& =\boldsymbol{\mathcal{P}}_{a}(\tilde{R}M)
\end{align*}
This means
\begin{equation}
\boldsymbol{\Upsilon}(\tilde{R}M)=\mathbf{vex}(\boldsymbol{\mathcal{P}}_{a}(\tilde{R}M))=\hat{R}\sum_{j=1}^{n_{{\rm R}}}\left(\frac{s_{j}}{2}\hat{\upsilon}_{j}^{a}\times\upsilon_{j}^{a}\right)\label{eq:SLAM_VEX_VM}
\end{equation}
Accordingly, $\tilde{R}M$ is equivalent to
\begin{equation}
\tilde{R}M=\hat{R}\sum_{j=1}^{n_{{\rm R}}}\left(s_{j}\upsilon_{j}^{a}\left(\upsilon_{j}^{r}\right)^{\top}\right)\label{eq:SLAM_RM_VM}
\end{equation}
Recall that ${\rm Tr}\left\{ M\right\} =3$. From the definition in
\eqref{eq:SLAM_Ecul_Dist}, one has
\begin{align}
E_{\tilde{R}}=||\tilde{R}M||_{{\rm I}} & =\frac{1}{4}{\rm Tr}\left\{ (\mathbf{I}_{3}-\tilde{R})M\right\} \nonumber \\
& =\frac{1}{4}{\rm Tr}\left\{ \mathbf{I}_{3}-\hat{R}\sum_{j=1}^{n_{{\rm R}}}\left(s_{j}\upsilon_{j}^{a}\left(\upsilon_{j}^{r}\right)^{\top}\right)\right\} \nonumber \\
& =\frac{1}{4}\sum_{j=1}^{n_{{\rm R}}}\left(1-s_{j}\left(\hat{\upsilon}_{j}^{a}\right)^{\top}\upsilon_{j}^{a}\right)\label{eq:SLAM_RI_VM}
\end{align}
Also, note that
\begin{align}
1-\left\Vert \tilde{R}\right\Vert _{{\rm I}} & =1-\frac{1}{4}{\rm Tr}\left\{ \mathbf{I}_{3}-\tilde{R}\right\} \nonumber \\
& =1-\frac{3}{4}+\frac{1}{4}{\rm Tr}\{\tilde{R}\}\nonumber \\
& =\frac{1}{4}\left(1+{\rm Tr}\{\tilde{R}\}\right)\label{eq:SLAM_Property}
\end{align}
One may rewrite the above results \eqref{eq:SLAM_Property} as
\begin{align}
1-||\tilde{R}||_{{\rm I}} & =\frac{1}{4}\left(1+{\rm Tr}\{\tilde{R}MM^{-1}\}\right)\label{eq:SLAM_property2}
\end{align}
In view of \eqref{eq:SLAM_M}, \eqref{eq:SLAM_RM_VM} and \eqref{eq:SLAM_property2},
one obtains
\begin{align}
& \pi(\tilde{R},M)={\rm Tr}\left\{ \tilde{R}MM^{-1}\right\} \nonumber \\
& \hspace{0.3em}={\rm Tr}\left\{ \left(\sum_{j=1}^{n_{{\rm R}}}s_{j}\upsilon_{j}^{a}\left(\upsilon_{j}^{r}\right)^{\top}\right)\left(\sum_{j=1}^{n_{{\rm R}}}s_{j}\hat{\upsilon}_{j}^{a}\left(\upsilon_{j}^{r}\right)^{\top}\right)^{-1}\right\} \label{eq:SLAM_Gamma_VM}
\end{align}
To this end, in the filter design it is considered that $\breve{\mathbf{M}}={\rm Tr}\left\{ M\right\} \mathbf{I}_{3}-M$,
$E_{\tilde{R}}=||\tilde{R}M||_{{\rm I}}$, $\pi(\tilde{R},M)$, $\boldsymbol{\Upsilon}(\tilde{R}M)$,
and $e_{i}$ are given relative to vector measurements as in \eqref{eq:SLAM_M},
\eqref{eq:SLAM_RI_VM}, \eqref{eq:SLAM_Gamma_VM}, \eqref{eq:SLAM_VEX_VM},
and \eqref{eq:SLAM_e_Final}, respectively, for all $i=1,2,\cdots,n$.
Consider the following nonlinear stochastic filter:

\begin{align}
\dot{\hat{\boldsymbol{T}}}= & \hat{\boldsymbol{T}}\left[U_{m}-\hat{b}_{U}-W_{U}\right]_{\wedge}\label{eq:SLAM_T_est_dot_f2}\\
\dot{{\rm \hat{p}}}_{i}= & -\frac{k_{2}}{\varrho}e_{i}+\hat{R}\left[y_{i}\right]_{\times}W_{\Omega},\hspace{1em}i=1,2,\ldots,n\label{eq:SLAM_p_est_dot_f2}\\
\dot{\hat{b}}_{U}= & \sum_{i=1}^{n}\frac{\Gamma}{\alpha_{i}}\left[\begin{array}{cc}
\frac{\alpha_{i}}{2}\tau_{b}\hat{R}^{\top} & -\left[y_{i}\right]_{\times}\hat{R}^{\top}\\
0_{3\times3} & -\hat{R}^{\top}
\end{array}\right]\left[\begin{array}{c}
\boldsymbol{\Upsilon}(\tilde{R}M)\\
\left\Vert e_{i}\right\Vert ^{2}e_{i}
\end{array}\right]\nonumber \\
& -k_{b}\Gamma\hat{b}_{U}\label{eq:SLAM_b_est_dot_f2}\\
\dot{\hat{\sigma}}= & \frac{\Gamma_{\sigma}}{8}\tau_{\sigma}{\rm diag}\left(\hat{R}^{\top}\boldsymbol{\Upsilon}(\tilde{R}M)\right)\hat{R}^{\top}\boldsymbol{\Upsilon}(\tilde{R}M)-k_{\sigma}\Gamma_{\sigma}\hat{\sigma}\label{eq:SLAM_s_est_dot_f2}\\
W_{U}= & \sum_{i=1}^{n}\frac{1}{\alpha_{i}}\left[\begin{array}{c}
\alpha_{i}\left(\frac{k_{1}}{\tau_{w}}\mathbf{I}_{3}+\frac{1}{4}\frac{E_{\tilde{R}}+2}{E_{\tilde{R}}+1}{\rm diag}\left(\hat{\sigma}\right)\right)\hat{R}^{\top}\boldsymbol{\Upsilon}(\tilde{R}M)\\
-k_{3}\hat{R}^{\top}\left\Vert e_{i}\right\Vert ^{2}e_{i}
\end{array}\right]\label{eq:SLAM_W_f2}
\end{align}
where $\tau_{b}=\left(E_{\tilde{R}}+1\right)\exp\left(E_{\tilde{R}}\right)$,
$\tau_{\sigma}=\left(E_{\tilde{R}}+2\right)\exp\left(E_{\tilde{R}}\right)$,
$\tau_{w}=\underline{\lambda}(\breve{\mathbf{M}})\left(1+\pi(\tilde{R},M)\right)$,
$W_{U}=\left[W_{\Omega}^{\top},W_{V}^{\top}\right]^{\top}\in\mathbb{R}^{6}$
is a correction factor, and $\hat{b}_{U}=\left[\hat{b}_{\Omega}^{\top},\hat{b}_{V}^{\top}\right]^{\top}\in\mathbb{R}^{6}$
is the estimate of $b_{U}$. $k_{1}$, $k_{2}$, $k_{3}$, $\Gamma_{\sigma}$,
$\Gamma=\left[\begin{array}{cc}
\Gamma_{1} & 0_{3\times3}\\
0_{3\times3} & \Gamma_{2}
\end{array}\right]$, and $\alpha_{i}$ are positive constants.
\begin{thm}
	Consider combining the stochastic SLAM dynamics $\dot{X}=\left(\dot{\boldsymbol{T}},\dot{\overline{{\rm p}}}\right)$
	in \eqref{eq:SLAM_True_dot_STOCH} with landmark measurements (output
	$\overline{y}_{i}=\boldsymbol{T}^{-1}\overline{{\rm p}}_{i}$) for
	all $i=1,2,\ldots,n$, inertial measurement units $\upsilon_{j}^{a}=R^{\top}\upsilon_{j}^{r}$
	for all $j=1,2,\ldots n_{{\rm R}}$, and velocity measurements ($U_{m}=U+b_{U}+n_{U}$)
	where $n_{U}\neq0$. Let Assumptions \ref{Assumption:Feature} and
	\ref{Assum:Boundedness} hold, and let the filter design be as in
	\eqref{eq:SLAM_T_est_dot_f2}, \eqref{eq:SLAM_p_est_dot_f2}, \eqref{eq:SLAM_b_est_dot_f2},
	\eqref{eq:SLAM_s_est_dot_f2}, and \eqref{eq:SLAM_W_f2}. Consider
	the design parameters $k_{2}>9/4$, $k_{b}$, $k_{\sigma}$, $k_{1}$,
	$k_{3}$, $\Gamma$, $\Gamma_{\sigma}$, and $\alpha_{i}$ to be positive
	constants and $\varrho$ to be sufficiently small. Consider the following
	set:
	\begin{align}
	\mathcal{S}= & \left\{ \left(\tilde{R},e_{1},e_{2},\ldots,e_{n}\right)\in\mathbb{SO}\left(3\right)\times\mathbb{R}^{3}\times\mathbb{R}^{3}\times\cdots\times\mathbb{R}^{3}\right|\nonumber \\
	& \hspace{9em}\tilde{R}=\mathbf{I}_{3},e_{i}=\underline{\mathbf{0}}_{3}\forall i=1,2,\ldots n\}\label{eq:SLAM_Set2}
	\end{align}
	Then, 1) all the closed loop error signals are SGUUB in mean square,
	and 2) the error $\left(\tilde{R},e_{1},e_{2},\ldots,e_{n}\right)$
	converges to the close neighborhood of $\mathcal{S}$ in probability
	for $\tilde{R}\left(0\right)\notin\mathcal{U}_{s}$.
\end{thm}
\begin{proof}Due to the fact that $\dot{\hat{\boldsymbol{T}}}$ in
	\eqref{eq:SLAM_T_est_dot_f2} is identical to \eqref{eq:SLAM_T_est_dot_f1},
	and in view of the pose error dynamics in \eqref{eq:SLAM_T_error_dot}
	one has
	\begin{align}
	de_{i}= & \left(\dot{\hat{{\rm p}}}_{i}-\left[\begin{array}{cc}
	-\hat{R}\left[y_{i}\right]_{\times} & \hat{R}\end{array}\right]\left(\tilde{b}_{U}-W_{U}\right)\right)dt\nonumber \\
	& -\left[\begin{array}{cc}
	-\hat{R}\left[y_{i}\right]_{\times} & \hat{R}\end{array}\right]\mathcal{Q}_{U}d\beta_{U}\nonumber \\
	= & \mathcal{F}_{i}dt+\mathcal{G}_{i}\mathcal{Q}_{U}d\beta_{U}\label{eq:SLAM_Attit_Error_dot-1}
	\end{align}
	where $\mathcal{F}_{i}=\dot{\hat{{\rm p}}}_{i}-\left[\begin{array}{cc}
	-\hat{R}\left[y_{i}\right]_{\times} & \hat{R}\end{array}\right]\left(\tilde{b}_{U}-W_{U}\right)$ and $\mathcal{G}_{i}=-\left[\begin{array}{cc}
	-\hat{R}\left[y_{i}\right]_{\times} & \hat{R}\end{array}\right]$. Also, the attitude error dynamics are
	\begin{align}
	\dot{\tilde{R}} & =\dot{\hat{R}}R^{\top}+\hat{R}\dot{R}^{\top}\nonumber \\
	& =\hat{R}\left[\tilde{b}_{\Omega}-W_{\Omega}\right]_{\times}R^{\top}+\hat{R}\left[\mathcal{Q}_{\Omega}\frac{d\beta_{\Omega}}{dt}\right]_{\times}R^{\top}\nonumber \\
	d\tilde{R} & =\left[\hat{R}(\tilde{b}_{\Omega}-W_{\Omega})\right]_{\times}\tilde{R}dt+\left[\hat{R}\mathcal{Q}_{\Omega}d\beta_{\Omega}\right]_{\times}\tilde{R}\label{eq:SLAM_Attit_Error_dot}
	\end{align}
	Recall the definition in \eqref{eq:SLAM_Ecul_Dist} where $E_{\tilde{R}}=||\tilde{R}M||_{{\rm I}}=\frac{1}{4}{\rm Tr}\left\{ (\mathbf{I}_{3}-\tilde{R})M\right\} $.
	Thereby, after considering the identity in \eqref{eq:SLAM_Identity6}
	one finds
	\begin{align}
	dE_{\tilde{R}}= & -\frac{1}{4}{\rm Tr}\left\{ \left[\hat{R}(\tilde{b}_{\Omega}-W_{\Omega})\right]_{\times}\tilde{R}M\right\} dt\nonumber \\
	& -\frac{1}{4}{\rm Tr}\left\{ \left[\hat{R}\mathcal{Q}_{\Omega}d\beta_{\Omega}\right]_{\times}\tilde{R}M\right\} \nonumber \\
	= & -\frac{1}{4}{\rm Tr}\left\{ \tilde{R}M\boldsymbol{\mathcal{P}}_{a}\left(\left[\hat{R}(\tilde{b}_{\Omega}-W_{\Omega})\right]_{\times}\right)\right\} dt\nonumber \\
	& -\frac{1}{4}{\rm Tr}\left\{ \tilde{R}M\boldsymbol{\mathcal{P}}_{a}\left(\left[\hat{R}\mathcal{Q}_{\Omega}d\beta_{\Omega}\right]_{\times}\right)\right\} \nonumber \\
	= & \frac{1}{2}\mathbf{vex}\left(\boldsymbol{\mathcal{P}}_{a}(\tilde{R}M)\right)^{\top}\hat{R}(\tilde{b}_{\Omega}-W_{\Omega})dt\nonumber \\
	& +\frac{1}{2}\mathbf{vex}\left(\boldsymbol{\mathcal{P}}_{a}(\tilde{R}M)\right)^{\top}\hat{R}\mathcal{Q}_{\Omega}d\beta_{\Omega}\nonumber \\
	= & fdt+g\mathcal{Q}_{\Omega}d\beta_{\Omega}\label{eq:SLAM_RI_VM_dot}
	\end{align}
	where $f=\frac{1}{2}\mathbf{vex}\left(\boldsymbol{\mathcal{P}}_{a}(\tilde{R}M)\right)^{\top}\hat{R}(\tilde{b}_{\Omega}-W_{\Omega})$
	and $g=\frac{1}{2}\mathbf{vex}\left(\boldsymbol{\mathcal{P}}_{a}(\tilde{R}M)\right)^{\top}\hat{R}$.
	It should be noted that $\dot{M}=0_{3\times3}$ according to the definition
	in \eqref{eq:SLAM_M}. Since $E_{\tilde{R}}=||\tilde{R}M||_{{\rm I}}$
	is greater than zero for all $||\tilde{R}M||_{{\rm I}}\neq0$ or equivalently
	$\tilde{R}\neq\mathbf{I}_{3}$ and $E_{\tilde{R}}=0$ only at $\tilde{R}=\mathbf{I}_{3}$,
	consider the candidate Lyapunov function $\boldsymbol{{\rm V}}=\boldsymbol{{\rm V}}\left(E_{\tilde{R}},e_{1},\ldots,e_{n},\tilde{b}_{U},\tilde{\sigma}\right)$\begin{strip}
		\begin{equation}
		\boldsymbol{{\rm V}}=\sum_{i=1}^{n}\frac{1}{4\alpha_{i}}\left\Vert e_{i}\right\Vert ^{4}+E_{\tilde{R}}\exp\left(E_{\tilde{R}}\right)+\frac{1}{2}\tilde{b}_{U}^{\top}\Gamma^{-1}\tilde{b}_{U}+\frac{1}{2}\tilde{\sigma}^{\top}\Gamma_{\sigma}^{-1}\tilde{\sigma}\label{eq:SLAM_Lyap2}
		\end{equation}
		The first and the second partial derivatives of \eqref{eq:SLAM_Lyap2}
		with respect to $E_{\tilde{R}}$ and $e_{i}$ are
		\begin{equation}
		\begin{cases}
		\boldsymbol{{\rm V}}_{E_{\tilde{R}}}=\partial\boldsymbol{{\rm V}}/\partial E_{\tilde{R}} & =\left(E_{\tilde{R}}+1\right)\exp\left(E_{\tilde{R}}\right)\\
		\boldsymbol{{\rm V}}_{E_{\tilde{R}}E_{\tilde{R}}}=\partial^{2}\boldsymbol{{\rm V}}/\partial E_{\tilde{R}}^{2} & =\left(E_{\tilde{R}}+2\right)\exp\left(E_{\tilde{R}}\right)
		\end{cases}\label{eq:SLAM_Vv_Vvv_RMI}
		\end{equation}
		\begin{equation}
		\begin{cases}
		\boldsymbol{{\rm V}}_{e_{i}}= & \partial\boldsymbol{{\rm V}}/\partial e_{i}=\left\Vert e_{i}\right\Vert ^{2}e_{i}\\
		\boldsymbol{{\rm V}}_{e_{i}e_{i}}= & \partial^{2}\boldsymbol{{\rm V}}/\partial e_{i}^{2}=\left\Vert e_{i}\right\Vert ^{2}\mathbf{I}_{3}+2e_{i}e_{i}^{\top}
		\end{cases}\label{eq:SLAM_Vv_Vvv_ei}
		\end{equation}
		From \eqref{eq:SLAM_Vv_Vvv_RMI} and \eqref{eq:SLAM_Vv_Vvv_ei}, the
		differential operator $\mathcal{L}\boldsymbol{{\rm V}}$ in Definition
		\ref{def:SLAM_LV} can be expressed as{\small{} }
		\begin{align}
		\mathcal{L}\boldsymbol{{\rm V}}= & \sum_{i=1}^{n}\left(\boldsymbol{{\rm V}}_{e_{i}}^{\top}\mathcal{F}_{i}+\frac{1}{2}{\rm Tr}\left\{ \mathcal{G}_{i}\mathcal{Q}_{U}^{2}\mathcal{G}_{i}^{\top}\boldsymbol{{\rm V}}_{e_{i}e_{i}}\right\} \right)+\boldsymbol{{\rm V}}_{E_{\tilde{R}}}^{\top}f\nonumber \\
		& +\frac{1}{2}{\rm Tr}\left\{ g\mathcal{Q}_{\Omega}^{2}g^{\top}\boldsymbol{{\rm V}}_{E_{\tilde{R}}E_{\tilde{R}}}\right\} -\tilde{b}_{U}^{\top}\Gamma^{-1}\dot{\hat{b}}_{U}-\tilde{\sigma}^{\top}\Gamma_{\sigma}^{-1}\dot{\hat{\sigma}}\nonumber \\
		= & \sum_{i=1}^{n}\frac{1}{\alpha_{i}}\left\Vert e_{i}\right\Vert ^{2}e_{i}^{\top}\left(\dot{\hat{{\rm p}}}_{i}-\left[\begin{array}{cc}
		-\hat{R}\left[y_{i}\right]_{\times} & \hat{R}\end{array}\right]\left(\tilde{b}_{U}-W_{U}\right)\right)\nonumber \\
		& +\frac{1}{2}{\rm Tr}\left\{ \sum_{i=1}^{n}\frac{1}{\alpha_{i}^{2}}\left[\begin{array}{cc}
		-\hat{R}\left[y_{i}\right]_{\times} & \hat{R}\end{array}\right]\mathcal{Q}_{U}^{2}\left[\begin{array}{cc}
		-\hat{R}\left[y_{i}\right]_{\times} & \hat{R}\end{array}\right]^{\top}\left(\left\Vert e_{i}\right\Vert ^{2}\mathbf{I}_{3}+2e_{i}e_{i}^{\top}\right)\right\} \nonumber \\
		& +\frac{1}{2}\tau_{b}\left[\begin{array}{cc}
		\boldsymbol{\Upsilon}(\tilde{R}M)^{\top}\hat{R} & \underline{\mathbf{0}}_{3}^{\top}\end{array}\right]\left(\tilde{b}_{U}-W_{U}\right)+\frac{1}{8}\tau_{\sigma}\boldsymbol{\Upsilon}(\tilde{R}M)^{\top}\hat{R}\mathcal{Q}_{\Omega}^{2}\hat{R}^{\top}\boldsymbol{\Upsilon}(\tilde{R}M)\nonumber \\
		& -\tilde{b}_{U}^{\top}\Gamma^{-1}\dot{\hat{b}}_{U}-\tilde{\sigma}\Gamma_{\sigma}^{-1}\dot{\hat{\sigma}}\label{eq:SLAM_Lyap2_dot2}
		\end{align}
		where $\tau_{b}=\left(E_{\tilde{R}}+1\right)\exp\left(E_{\tilde{R}}\right)$
		and $\tau_{\sigma}=\left(E_{\tilde{R}}+2\right)\exp\left(E_{\tilde{R}}\right)$
		with both $\tau_{b}$ and $\tau_{\sigma}$ being positive for all
		$t\geq0$. One can easily find that
		\begin{align*}
		{\rm Tr}\left\{ \left\Vert e_{i}\right\Vert ^{2}\mathbf{I}_{3}+2e_{i}e_{i}^{\top}\right\}  & \leq{\rm Tr}\left\{ 3\left\Vert e_{i}\right\Vert ^{2}\mathbf{I}_{3}\right\} \\
		& \leq9\left\Vert e_{i}\right\Vert ^{2}
		\end{align*}
		Based on the result above and by the virtue of Young's inequality
		in Lemma \ref{lem:SLAM_Young}, one finds
		\begin{align}
		& \frac{9}{2}{\rm Tr}\left\{ \sum_{i=1}^{n}\frac{1}{\alpha_{i}^{2}}\left[\begin{array}{cc}
		-\hat{R}\left[y_{i}\right]_{\times} & \hat{R}\end{array}\right]\mathcal{Q}_{U}^{2}\left[\begin{array}{cc}
		-\hat{R}\left[y_{i}\right]_{\times} & \hat{R}\end{array}\right]^{\top}\right\} \left\Vert e_{i}\right\Vert ^{2}\leq\nonumber \\
		& \hspace{5em}\leq\sum_{i=1}^{n}\frac{9}{4\alpha_{i}^{2}\varrho}\left\Vert e_{i}\right\Vert ^{4}+\sum_{i=1}^{n}\frac{9\varrho}{4\alpha_{i}^{2}}{\rm Tr}\left\{ \hat{R}\left(\left[y_{i}\right]_{\times}^{\top}{\rm diag}\left(\sigma\right)\left[y_{i}\right]_{\times}+{\rm diag}\left(\sigma\right)\right)\hat{R}^{\top}\right\} ^{2}\label{eq:SLAM_property9}
		\end{align}
		Note that $\sigma$ is the upper bound of $\mathcal{Q}_{U}^{2}$.
		Due to the fact that $\hat{R}\in\mathbb{SO}\left(3\right)$ which
		is orthogonal, for any $M\in\mathbb{R}^{3}$ one has ${\rm Tr}\left\{ \hat{R}M\hat{R}^{\top}\right\} ={\rm Tr}\left\{ M\right\} $.
		Also, let $\mathbf{y_{i}}$ be the upper bound of $y_{i}$. As such,
		define
		\begin{equation}
		c_{2}=\sum_{i=1}^{n}\frac{9\varrho}{4\alpha_{i}^{2}}{\rm Tr}\left\{ -\left[\mathbf{y_{i}}\right]_{\times}{\rm diag}\left(\sigma\right)\left[\mathbf{y_{i}}\right]_{\times}+{\rm diag}\left(\sigma\right)\right\} ^{2}\label{eq:SLAM_Property10}
		\end{equation}
		From \eqref{eq:SLAM_property9} and \eqref{eq:SLAM_Property10}, one
		may express \eqref{eq:SLAM_Lyap2_dot2} in inequality form as
		\begin{align}
		\mathcal{L}\boldsymbol{{\rm V}}\leq & \sum_{i=1}^{n}\frac{1}{\alpha_{i}}\left\Vert e_{i}\right\Vert ^{2}e_{i}^{\top}\dot{\hat{{\rm p}}}_{i}+\sum_{i=1}^{n}\frac{1}{\alpha_{i}}\left[\begin{array}{c}
		\boldsymbol{\Upsilon}(\tilde{R}M)\\
		\left\Vert e_{i}\right\Vert ^{2}e_{i}
		\end{array}\right]^{\top}\left[\begin{array}{c|c}
		\frac{\alpha_{i}}{2}\tau_{b}\hat{R} & 0_{3\times3}\\
		\hline \hat{R}\left[y_{i}\right]_{\times} & -\hat{R}
		\end{array}\right]\left(\tilde{b}_{U}-W_{U}\right)\nonumber \\
		& +\frac{1}{8}\tau_{\sigma}\boldsymbol{\Upsilon}\left(\tilde{R}M\right)^{\top}\hat{R}\left[\hat{R}^{\top}\boldsymbol{\Upsilon}\left(\tilde{R}M\right)\right]_{{\rm D}}\sigma+\sum_{i=1}^{n}\frac{9}{4\alpha_{i}^{2}\varrho}\left\Vert e_{i}\right\Vert ^{4}+c_{2}\nonumber \\
		& -\tilde{b}_{U}^{\top}\Gamma_{b}^{-1}\dot{\hat{b}}_{U}-\tilde{\sigma}^{\top}\Gamma_{\sigma}^{-1}\dot{\hat{\sigma}}\label{eq:SLAM_Lyap2_dot3}
		\end{align}
		Directly substituting $W_{U}$, $\dot{\hat{b}}_{U}$, $\dot{\hat{\sigma}}$
		and $\dot{\hat{{\rm p}}}_{i}$ with the definitions in \eqref{eq:SLAM_p_est_dot_f2},
		\eqref{eq:SLAM_b_est_dot_f2}, \eqref{eq:SLAM_s_est_dot_f2}, and
		\eqref{eq:SLAM_W_f2}, respectively, one shows
		\begin{align}
		\mathcal{L}\boldsymbol{{\rm V}}\leq & -\sum_{i=1}^{n}\frac{k_{2}\alpha_{i}-9/4}{\varrho\alpha_{i}^{2}}\left\Vert e_{i}\right\Vert ^{4}-k_{3}\left\Vert \sum_{i=1}^{n}\frac{\left\Vert e_{i}\right\Vert ^{2}}{\alpha_{i}}\hat{R}e_{i}\right\Vert ^{2}-\frac{k_{1}\tau_{b}}{2\tau_{w}}\left\Vert \boldsymbol{\Upsilon}\left(\tilde{R}M\right)\right\Vert ^{2}+k_{b}\tilde{b}_{U}^{\top}\hat{b}_{U}\nonumber \\
		& +k_{\sigma}\tilde{\sigma}^{\top}\hat{\sigma}+c_{2}\label{eq:SLAM_Lyap2_dot4}
		\end{align}
		As a result of \eqref{eq:SLAM_lemm1_2} in Lemma \ref{Lemm:SLAM_RM_I2},
		one has
		\begin{equation}
		\frac{2}{\tau_{w}}||\boldsymbol{\Upsilon}\left(\tilde{R}M\right)||^{2}=\frac{2}{\tau_{w}}\left\Vert \boldsymbol{\Upsilon}\left(\tilde{R}M\right)\right\Vert ^{2}\geq\left\Vert \tilde{R}M\right\Vert _{{\rm I}}\label{eq:SLAM_property11}
		\end{equation}
		where $E_{\tilde{R}}=||\tilde{R}M||_{{\rm I}}$ and $\tau_{w}=\underline{\lambda}(\breve{\mathbf{M}})\left(1+\pi(\tilde{R},M)\right)=\underline{\lambda}(\breve{\mathbf{M}})\left(1+{\rm Tr}\left\{ \tilde{R}MM^{-1}\right\} \right)$.
		Also
		\begin{align*}
		k_{b}\tilde{b}_{U}^{\top}b_{U} & \leq\frac{k_{b}}{2}\left\Vert b_{U}\right\Vert ^{2}+\frac{k_{b}}{2}\left\Vert \tilde{b}_{U}\right\Vert ^{2}\\
		k_{\sigma}\tilde{\sigma}^{\top}\sigma & \leq\frac{k_{\sigma}}{2}\left\Vert \sigma\right\Vert ^{2}+\frac{k_{\sigma}}{2}\left\Vert \tilde{\sigma}\right\Vert ^{2}
		\end{align*}
		According to the above result and \eqref{eq:SLAM_property11}, one
		can rewrite \eqref{eq:SLAM_Lyap2_dot4} as
		\begin{align}
		\mathcal{L}\boldsymbol{{\rm V}}\leq & -\sum_{i=1}^{n}\frac{k_{2}\alpha_{i}-9/4}{\varrho\alpha_{i}^{2}}\left\Vert e_{i}\right\Vert ^{4}-k_{3}\left\Vert \sum_{i=1}^{n}\frac{\left\Vert e_{i}\right\Vert ^{2}}{\alpha_{i}}\hat{R}e_{i}\right\Vert ^{2}-\frac{k_{1}}{4}\left(E_{\tilde{R}}+1\right)\exp\left(E_{\tilde{R}}\right)E_{\tilde{R}}\nonumber \\
		& -\frac{k_{b}}{2}\left\Vert \tilde{b}_{U}\right\Vert ^{2}-\frac{k_{\sigma}}{2}\left\Vert \tilde{\sigma}\right\Vert ^{2}+\frac{k_{b}}{2}\left\Vert b_{U}\right\Vert ^{2}+\frac{k_{\sigma}}{2}\left\Vert \sigma\right\Vert ^{2}+c_{2}\label{eq:SLAM_Lyap2_final1}
		\end{align}
		Define
		\begin{align*}
		\tau_{i}= & \frac{k_{2}\alpha_{i}-9/4}{\varrho\alpha_{i}^{2}},\hspace{1em}\forall i=1,2,\ldots,n\\
		\eta_{2}= & \frac{k_{b}}{2}\left\Vert b_{U}\right\Vert ^{2}+\frac{k_{\sigma}}{2}\left\Vert \sigma\right\Vert ^{2}+c_{2}
		\end{align*}
		Also, let
		\begin{align*}
		\mathcal{H} & =\left[\begin{array}{c|c}
		\begin{array}{ccc}
		\tau_{1}\mathbf{I}_{4} & \cdots & 0_{4\times4}\\
		\vdots & \ddots & \vdots\\
		0_{4\times4} & \cdots & \tau_{n}\mathbf{I}_{4}
		\end{array} & 0_{4n\times10}\\
		\hline 0_{10\times4n} & \begin{array}{ccc}
		\frac{k_{1}}{4} & 0_{1\times6} & 0_{1\times3}\\
		0_{6\times1} & \frac{1}{2}k_{b}\Gamma & 0_{6\times3}\\
		0_{3\times1} & 0_{3\times6} & \frac{1}{2}k_{\sigma}\Gamma_{\sigma}
		\end{array}
		\end{array}\right]\\
		\tilde{Y} & =\left[\begin{array}{c|c}
		\begin{array}{ccc}
		\frac{1}{4\alpha_{1}}\mathbf{I}_{4} & \cdots & 0_{4\times4}\\
		\vdots & \ddots & \vdots\\
		0_{4\times4} & \cdots & \frac{1}{4\alpha_{n}}\mathbf{I}_{4}
		\end{array} & 0_{4n\times10}\\
		\hline 0_{10\times4n} & \begin{array}{ccc}
		1 & 0_{1\times6} & 0_{1\times3}\\
		0_{6\times1} & \frac{1}{2}\Gamma^{-1} & 0_{6\times3}\\
		0_{3\times1} & 0_{3\times6} & \frac{1}{2}\Gamma_{\sigma}^{-1}
		\end{array}
		\end{array}\right]^{1/2}\left[\begin{array}{c}
		\frac{\left\Vert e_{1}\right\Vert }{2\sqrt{\alpha_{1}}}e_{1}\\
		\vdots\\
		\frac{\left\Vert e_{n}\right\Vert }{2\sqrt{\alpha_{n}}}e_{n}\\
		\begin{array}{c}
		\sqrt{E_{\tilde{R}}\exp\left(E_{\tilde{R}}\right)}\\
		\tilde{b}_{U}\\
		\tilde{\sigma}
		\end{array}
		\end{array}\right]
		\end{align*}
	\end{strip}Therefore, the differential operator in \eqref{eq:SLAM_Lyap2_final1}
	becomes
	\begin{equation}
	\mathcal{L}\boldsymbol{{\rm V}}\leq-\underline{\lambda}\left(\mathcal{H}\right)\boldsymbol{{\rm V}}+\eta_{2}\label{eq:SLAM_Lyap2_final2}
	\end{equation}
	where $\underline{\lambda}\left(\cdot\right)$ represents the minimum
	eigenvalue of a matrix. Based on \eqref{eq:SLAM_Lyap2_final2}, one
	finds
	\begin{equation}
	\frac{d\left(\mathbb{E}\left[\boldsymbol{{\rm V}}\right]\right)}{dt}=\mathbb{E}\left[\mathcal{L}\boldsymbol{{\rm V}}\right]\leq-\underline{\lambda}\left(\mathcal{H}\right)\mathbb{E}\left[\boldsymbol{{\rm V}}\right]+\eta_{2}\label{eq:SLAM_Lyap2_final3}
	\end{equation}
	Let $c=\mathbb{E}\left[\boldsymbol{{\rm V}}\left(t\right)\right]$;
	hence $\frac{d\left(\mathbb{E}\left[\boldsymbol{{\rm V}}\right]\right)}{dt}\leq0$
	for $\underline{\lambda}\left(\mathcal{H}\right)>\frac{\eta_{2}}{c}$.
	As such, $\boldsymbol{{\rm V}}\leq c$ is an invariant set and for
	$\mathbb{E}\left[\boldsymbol{{\rm V}}\left(0\right)\right]\leq c$
	there is $\mathbb{E}\left[\boldsymbol{{\rm V}}\left(t\right)\right]\leq c\forall t>0$.
	In view of Lemma \ref{Lemm:SLAM_deng}, the inequality in \eqref{eq:SLAM_Lyap2_final3}
	holds for $\boldsymbol{{\rm V}}\left(0\right)\leq c$ and for all
	$t>0$ such that
	\begin{equation}
	0\leq\mathbb{E}\left[\boldsymbol{{\rm V}}\left(t\right)\right]\leq\boldsymbol{{\rm V}}\left(0\right){\rm exp}\left(-\underline{\lambda}\left(\mathcal{H}\right)t\right)+\frac{\eta_{2}}{\underline{\lambda}\left(\mathcal{H}\right)},\,\forall t\geq0\label{eq:SLAM_Lyap2_final4}
	\end{equation}
	Hence, $\mathbb{E}\left[V\left(t\right)\right]$ is eventually bounded
	by $\eta_{2}/\underline{\lambda}\left(\mathcal{H}\right)$ in turn
	implying $\tilde{Y}$ is SGUUB in the mean square. Therefore, the
	result in \eqref{eq:SLAM_Lyap2_final1} guarantees that $e_{i}$ as
	well as $\tilde{R}$ are regulated to the neighborhood of the set
	$\mathcal{S}$ defined in \eqref{eq:SLAM_Set2} for all $i=1,2,\ldots,n$
	and $\tilde{R}\left(0\right)\notin\mathcal{U}_{s}$. In addition,
	$\tilde{P}\rightarrow P_{c}$ as $t\rightarrow\infty$. This completes
	the proof.\end{proof}

Algorithm \ref{alg:Alg1} details the implementation stages of the
nonlinear stochastic filter for SLAM defined in in \eqref{eq:SLAM_T_est_dot_f2},
\eqref{eq:SLAM_p_est_dot_f2}, \eqref{eq:SLAM_b_est_dot_f2}, \eqref{eq:SLAM_s_est_dot_f2},
and \eqref{eq:SLAM_W_f2}. 

\begin{algorithm}
	\caption{\label{alg:Alg1}Implementation steps of the nonlinear stochastic
		filter for SLAM}
	
	\textbf{Initialization}:
	\begin{enumerate}
		\item[{\footnotesize{}1:}] Set $\hat{R}\left(0\right)\in\mathbb{SO}\left(3\right)$ and $\hat{P}\left(0\right)\in\mathbb{R}^{3}$.
		Instead, establish $\hat{R}\left(0\right)\in\mathbb{SO}\left(3\right)$
		using any method of attitude determination, see \cite{hashim2020AtiitudeSurvey}
		\item[{\footnotesize{}2:}] Set ${\rm \hat{p}}_{i}\left(0\right)\in\mathbb{R}^{3}$ for all $i=1,2,\ldots,n$
		\item[{\footnotesize{}3:}] Set $\hat{b}_{\Omega}\left(0\right),,\hat{b}_{V}\left(0\right),\hat{\sigma}\left(0\right)\in\mathbb{R}^{3}$
		\item[{\footnotesize{}4:}] Select $k_{w}$, $k_{w2}$, $k_{p}$, $\Gamma$, $\Gamma_{\sigma}$,
		$\varrho$, $k_{b}$, $k_{\sigma}$, and $\alpha_{i}$ as positive
		constants
	\end{enumerate}
	\textbf{while }(1)\textbf{ do}
	\begin{enumerate}
		\item[{\footnotesize{}5:}] \textbf{for} $j=1:n_{{\rm R}}$
		\item[{\footnotesize{}6:}] \hspace{0.5cm}Measurements and observations as in \eqref{eq:SLAM_Vect_R}
		\item[{\footnotesize{}7:}] \hspace{0.5cm} $\upsilon_{j}^{r}=\frac{r_{j}}{\left\Vert r_{j}\right\Vert },\upsilon_{j}^{a}=\frac{a_{j}}{\left\Vert a_{j}\right\Vert }$
		as in \eqref{eq:SLAM_Vector_norm}
		\item[{\footnotesize{}8:}] \hspace{0.5cm}$\hat{\upsilon}_{j}^{a}=\hat{R}^{\top}\upsilon_{j}^{r}$
		as in \eqref{eq:SLAM_vect_R_estimate}
		\item[{\footnotesize{}9:}] \textbf{end for}
		\item[{\footnotesize{}10:}] $M=\sum_{j=1}^{n_{{\rm R}}}s_{j}\upsilon_{j}^{r}\left(\upsilon_{j}^{r}\right)^{\top}$
		as in \eqref{eq:SLAM_M} with $\breve{\mathbf{M}}={\rm Tr}\left\{ M\right\} \mathbf{I}_{3}-M$
		\item[{\footnotesize{}11:}] $\boldsymbol{\Upsilon}=\hat{R}\sum_{j=1}^{n_{{\rm R}}}\left(\frac{s_{j}}{2}\hat{\upsilon}_{j}^{a}\times\upsilon_{j}^{a}\right)$
		as in \eqref{eq:SLAM_VEX_VM}
		\item[{\footnotesize{}12:}] $\pi={\rm Tr}\left\{ \left(\sum_{j=1}^{n_{{\rm R}}}s_{j}\upsilon_{j}^{a}\left(\upsilon_{j}^{r}\right)^{\top}\right)\left(\sum_{j=1}^{n_{{\rm R}}}s_{j}\hat{\upsilon}_{j}^{a}\left(\upsilon_{j}^{r}\right)^{\top}\right)^{-1}\right\} $
		as in \eqref{eq:SLAM_Gamma_VM}
		\item[{\footnotesize{}13:}] \textbf{for} $i=1:n$
		\item[{\footnotesize{}14:}] \hspace{0.5cm}$e_{i}=\hat{{\rm p}}_{i}-\hat{R}y_{i}-\hat{P}$ as
		in \eqref{eq:SLAM_e_Final}
		\item[{\footnotesize{}15:}] \textbf{end for}
		\item[{\footnotesize{}16:}] $W_{\Omega}=\left(\frac{k_{1}}{\tau_{w}}\mathbf{I}_{3}+\frac{1}{4}\frac{E_{\tilde{R}}+2}{E_{\tilde{R}}+1}{\rm diag}\left(\hat{\sigma}\right)\right)\hat{R}^{\top}\boldsymbol{\Upsilon}$,
		with $\tau_{R}=\underline{\lambda}(\breve{\mathbf{M}})\times(1+\pi)$
		\item[{\footnotesize{}17:}] $W_{V}=-\sum_{i=1}^{n}\frac{k_{3}}{\alpha_{i}}\left\Vert e_{i}\right\Vert ^{2}\hat{R}^{\top}e_{i}$
		\item[{\footnotesize{}18:}] $\dot{\hat{R}}=\hat{R}\left[\Omega_{m}-\hat{b}_{\Omega}-W_{\Omega}\right]_{\times}$
		\item[{\footnotesize{}19:}] $\dot{\hat{P}}=\hat{R}\left(V_{m}-\hat{b}_{V}-W_{V}\right)$
		\item[{\footnotesize{}20:}] \textbf{for} $i=1:n$
		\item[{\footnotesize{}21:}] \hspace{0.5cm}$\dot{{\rm \hat{p}}}_{i}=-\frac{k_{2}}{\varrho}e_{i}+\hat{R}\left[y_{i}\right]_{\times}W_{\Omega}$
		\item[{\footnotesize{}22:}] \textbf{end for}
		\item[{\footnotesize{}23:}] $\dot{\hat{b}}_{\Omega}=\frac{\Gamma_{1}}{2}\tau_{b}\hat{R}^{\top}\boldsymbol{\Upsilon}-\sum_{i=1}^{n}\frac{\Gamma_{1}}{\alpha_{i}}\left\Vert e_{i}\right\Vert ^{2}\left[y_{i}\right]_{\times}\hat{R}^{\top}e_{i}-k_{b}\Gamma_{1}\hat{b}_{\Omega}$
		\item[{\footnotesize{}24:}] $\dot{\hat{b}}_{V}=-\sum_{i=1}^{n}\frac{\Gamma_{2}}{\alpha_{i}}\left\Vert e_{i}\right\Vert ^{2}\hat{R}^{\top}e_{i}-k_{b}\Gamma_{2}\hat{b}_{V}$
		\item[{\footnotesize{}25:}] $\dot{\hat{\sigma}}=\frac{\Gamma_{\sigma}}{8}\tau_{\sigma}{\rm diag}\left(\hat{R}^{\top}\boldsymbol{\Upsilon}\right)\hat{R}^{\top}\boldsymbol{\Upsilon}-k_{\sigma}\Gamma_{\sigma}\hat{\sigma}$
	\end{enumerate}
	\textbf{end while}
\end{algorithm}

\section{Numerical Results \label{sec:SE3_Simulations}}

\subsection{Simulation\label{subsec:Simulation}}

This section demonstrates the robustness of the proposed stochastic
estimator for SLAM on $\mathbb{SLAM}_{n}\left(3\right)$ Lie group.
Consider the true angular and translational velocities of the vehicle
in 3D space to be $\Omega=[0,0,0.3]^{\top}({\rm rad/sec})$ and $V=[2.5,0,0]^{\top}({\rm m/sec})$,
respectively. Also, let its true initial pose be 
\[
R\left(0\right)=\mathbf{I}_{3},\hspace{1em}P\left(0\right)=[0,0,3]^{\top}
\]
Let four landmarks be fixed with respect to $\left\{ \mathcal{I}\right\} $
within the unknown environment and positioned at ${\rm p}_{1}=[6,0,0]^{\top}$,
${\rm p}_{2}=[-6,0,0]^{\top}$, ${\rm p}_{3}=[0,6,0]^{\top}$, and
${\rm p}_{4}=[0,-6,0]^{\top}$. Let angular and translational velocities
be corrupted by unknown constant bias $b_{U}=\left[b_{\Omega}^{\top},b_{V}^{\top}\right]^{\top}$
with $b_{\Omega}=[0.1,-0.1,-0.1]^{\top}({\rm rad/sec})$ and $b_{V}=[0.08,0.07,-0.06]^{\top}({\rm m/sec})$,
respectively. Additionally, assume that the group velocity vector
is altered by unknown noise $n_{U}=\left[n_{\Omega}^{\top},n_{V}^{\top}\right]^{\top}$
with $n_{\Omega}=\mathcal{N}\left(0,0.2\right)({\rm rad/sec})$ and
$n_{V}=\mathcal{N}\left(0,0.2\right)({\rm m/sec})$. It should be
noted that abbreviation $n_{\Omega}=\mathcal{N}\left(0,0.2\right)$
indicates that the random noise vector $n_{\Omega}$ is normally distributed
around a zero mean with a standard deviation of $0.2$. Consider two
non-collinear inertial-frame observations equal to $r_{1}=\left[-1,1,1.1\right]^{\top}$
and $r_{2}=\left[0,0,1.3\right]^{\top}$ where the associated body-frame
measurements are defined as in \eqref{eq:SLAM_Vect_R}. As was indicated
by Remarks \ref{rem:R_Marix}, the third observation and measurement
can be calculated using a cross product of the two available observations.
To account for large error in initialization, the initial estimate
of attitude and position are set as
\begin{align*}
\hat{R}\left(0\right) & =\left[\begin{array}{ccc}
0.8090 & -0.5878 & 0\\
0.5878 & 0.8090 & 0\\
0 & 0 & 1
\end{array}\right],\hspace{1em}\hat{P}\left(0\right)=[0,0,0]^{\top}
\end{align*}
The four landmark estimates are initiated at positions: $\hat{{\rm p}}_{1}\left(0\right)=\hat{{\rm p}}_{2}\left(0\right)=\hat{{\rm p}}_{3}\left(0\right)=\hat{{\rm p}}_{4}\left(0\right)=[0,0,0]^{\top}$.
The design parameters are selected as $\alpha_{i}=0.05$, $\Gamma_{1}=3\mathbf{I}_{3}$,
$\Gamma_{2}=10000\mathbf{I}_{3}$, $\Gamma_{\sigma}=10$, $k_{1}=10$,
$k_{2}=10$, $k_{3}=10$, $k_{\sigma}=0.02$, and $\varrho=0.5$ while
the initial values of bias and covariance estimates are $\hat{b}_{U}\left(0\right)=\underline{\mathbf{0}}_{6}$
and $\hat{\sigma}\left(0\right)=\underline{\mathbf{0}}_{3}$, respectively,
for all $i=1,2,3,4$. Also, select $k_{b}=10^{-13}$ as a very small
constant.

Figure \ref{fig:SLAM_3d} highlights the contrast between the true
and measured values of angular and translational velocities. The evolution
of estimate trajectories output by the proposed SLAM nonlinear stochastic
filter is depicted in Figure \ref{fig:SLAM_3d}. As demonstrated by
Figure \ref{fig:SLAM_3d}, despite large initialization error, the
robot's position converged smoothly and continuously from the zero
point of origin to the true trajectory of travel arriving at the desired
terminal point. Analogously, landmark estimates, initiated at the
origin, rapidly diverged to their true locations.

\begin{figure}[h!]
	\centering{}\includegraphics[scale=0.29]{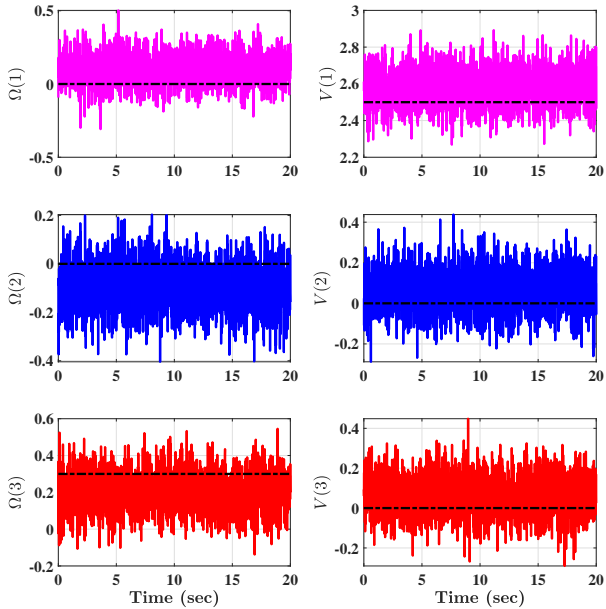}\caption{Angular and translational velocities: measured plotted in colored
		solid-line vs true plotted in black center-line.}
	\label{fig:SLAM_Vel}
\end{figure}

\begin{figure}[h]
	\centering{}\includegraphics[scale=0.29]{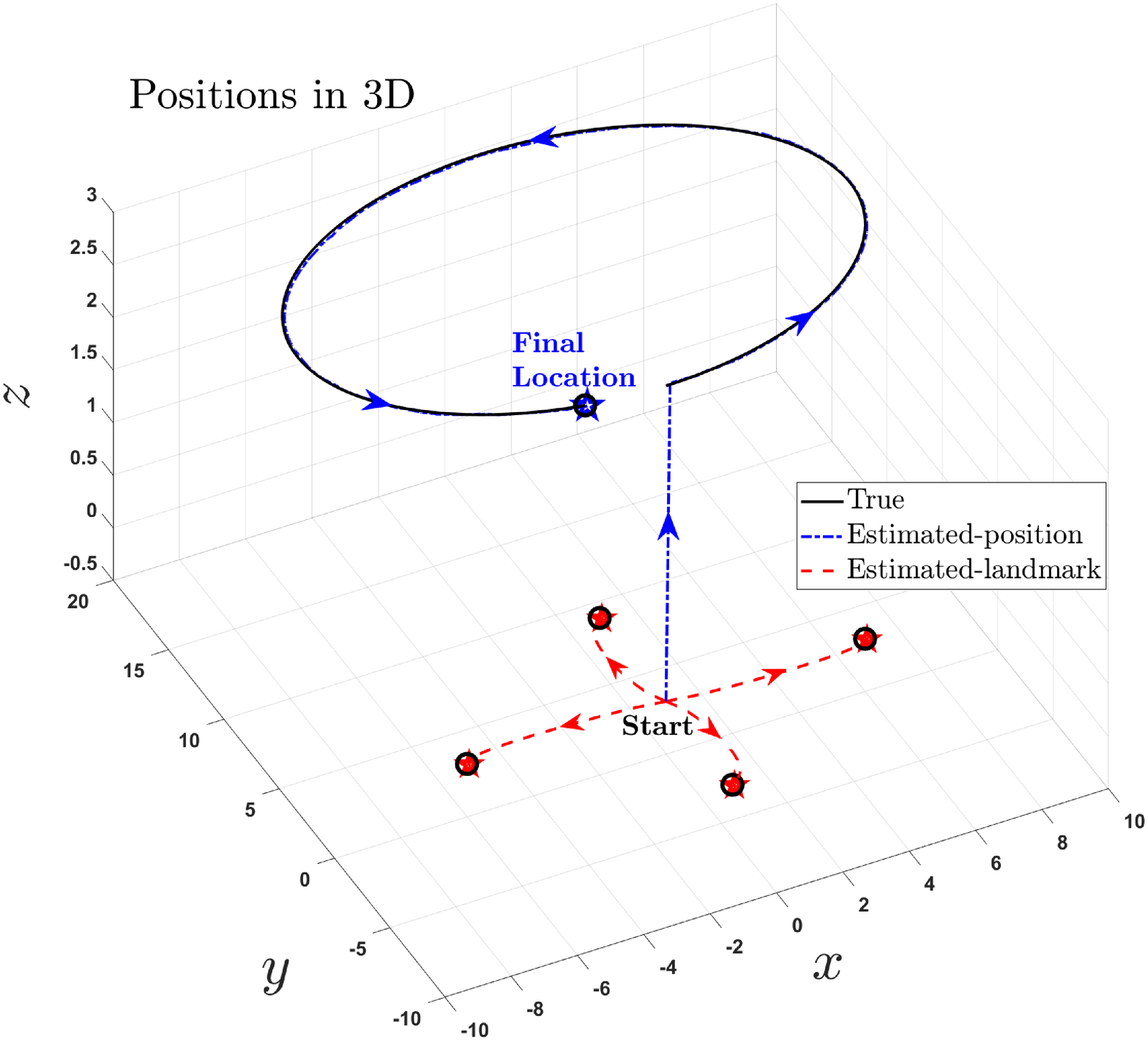}\caption{Output performance of the proposed nonlinear stochastic filter for
		SLAM described in Subsection \ref{subsec:Det_with_IMU} and detailed
		in Algorithm \ref{alg:Alg1} plotted against the true robot's position
		and landmark locations in 3D space. The true robot trajectory is plotted
		in black solid-line with the black circle marking its terminal point.
		The black circles also mark the true fixed landmarks. Estimation of
		the robot's position is plotted as a blue center-line initiating at
		the origin and converging to its final location marked with a blue
		star $\star$. Landmark estimation trajectories depicted as red dashed-lines
		initiate at $(0,0,0)$ and diverge to their final positions marked
		with red stars $\star$.}
	\label{fig:SLAM_3d}
\end{figure}

\begin{figure}[h!]
	\centering{}\includegraphics[scale=0.3]{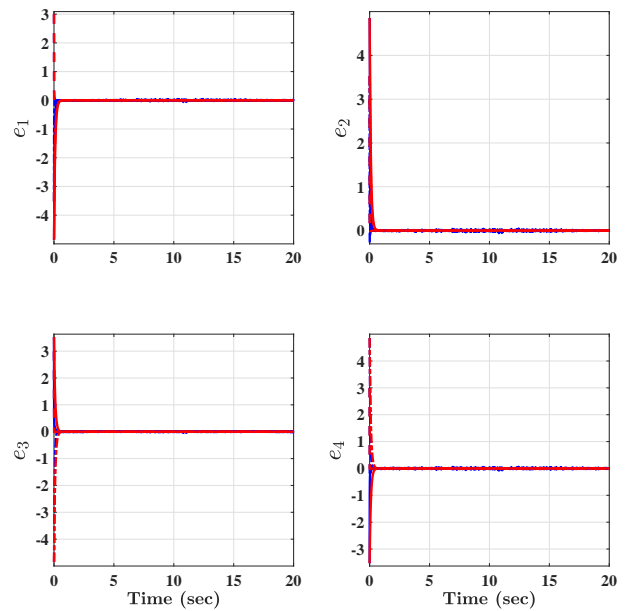}\caption{Error trajectories utilized in the Lyapunov function candidate. The
		proposed nonlinear stochastic estimator for SLAM with IMU outlined
		in Subsection \ref{subsec:Det_with_IMU} is depicted in blue against
		the deterministic nonlinear estimator for SLAM presented in Subsection
		\ref{subsec:Det_without_IMU} depicted in red.}
	\label{fig:SLAM_error_Lyap}
\end{figure}

The asymptotic convergence of the error trajectories of $e_{i}$ achieved
by the nonlinear filter for SLAM with IMU (stochastic) and without
IMU (deterministic) is demonstrated in Figure \ref{fig:SLAM_error_Lyap}
for all $i=1,2,3,4$. Consider the error defined as $||\tilde{R}||_{{\rm I}}=\frac{1}{4}{\rm Tr}\{\mathbf{I}_{3}-\tilde{R}\}$
where $\tilde{R}=\hat{R}R^{\top}$, $\tilde{P}=\hat{P}-\tilde{R}P$,
and $\tilde{{\rm p}}_{i}=\hat{{\rm p}}_{i}-\tilde{R}{\rm p}_{i}$.
It is apparent that $e_{i}=\tilde{{\rm p}}_{i}-\tilde{P}$ does not
necessarily result in $||\tilde{R}||_{{\rm I}}\rightarrow0$, $\tilde{P}\rightarrow0$,
and $\tilde{{\rm p}}_{i}\rightarrow0$. When designing a SLAM filter,
convergence of $\tilde{R}$,$\tilde{P}$, and $\tilde{{\rm p}}_{i}$
to a constant does not constitute the ultimate goal. The true objective
is to drive $||\tilde{R}||_{{\rm I}}\rightarrow0$, $||P-\hat{P}||\rightarrow0$,
and $||{\rm p}_{i}-{\rm \hat{p}}_{i}||\rightarrow0$. As such, Figure
\ref{fig:SLAM_error} benchmarks the output performance of the proposed
stochastic estimator for SLAM with IMU highlighting its superiority
over the deterministic solution without IMU. Actually, IMU facilitates
achieving $\tilde{R}\rightarrow\mathbf{I}_{3}$ which in turn leads
to $||\tilde{R}||_{{\rm I}}\rightarrow0$ as $t\rightarrow\infty$
significantly reducing error values of $||P-\hat{P}||$ and $||{\rm p}_{i}-{\rm \hat{p}}_{i}||$.
This indeed is true as $\tilde{P}=\hat{P}-\tilde{R}P$, and $\tilde{{\rm p}}_{i}=\hat{{\rm p}}_{i}-\tilde{R}{\rm p}_{i}$
causing $\tilde{R}$ to strongly influence the values of $\tilde{P}$
and $\tilde{{\rm p}}_{i}$. Figure \ref{fig:SLAM_error} reveals the
robustness of the proposed stochastic estimator for SLAM using IMU.
Figure \ref{fig:SLAM_error} illustrating its strong convergence as
well as tracking capabilities. In contrast, as can be clearly seen
in Figure \ref{fig:SLAM_error}, the deterministic nonlinear estimator
without IMU shows unreasonable performance in agreement with \cite{hashim2020LetterSLAM,zlotnik2018SLAM}.
It should be noted that presence of the residual error is unavoidable
for $||P-\hat{P}||$ and $||{\rm p}_{i}-{\rm \hat{p}}_{i}||$ for
the proposed stochastic filter illustrated by Figure \ref{fig:SLAM_error_Steady_State}.
Nonetheless, the nonlinear stochastic filter proposed in Subsection
\ref{subsec:Det_with_IMU} outperforms the nonlinear deterministic
filter presented in Subsection \ref{subsec:Det_without_IMU} in terms
of the convergence rate of $||\tilde{R}||_{{\rm I}}$ and $||P-\hat{P}||$
by a wide margin.

\begin{figure}[h!]
	\centering{}\includegraphics[scale=0.24]{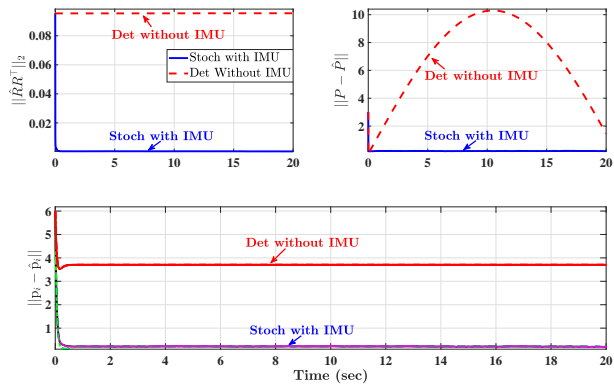}\caption{Output performance of $||\hat{R}R^{\top}||_{{\rm I}}$, $||P-\hat{P}||$
		and $||{\rm p}_{i}-{\rm \hat{p}}_{i}||$ for all $i=1,2,3,4$. All
		colors other than red represent the proposed nonlinear stochastic
		filter based on velocity, landmark, and IMU measurements, while red
		represents the nonlinear filter based only on velocity and landmark
		measurements. Det and Stoch abbreviate deterministic and stochastic
		filters, respectively.}
	\label{fig:SLAM_error}
\end{figure}

\begin{figure}[h!]
	\centering{}\includegraphics[scale=0.29]{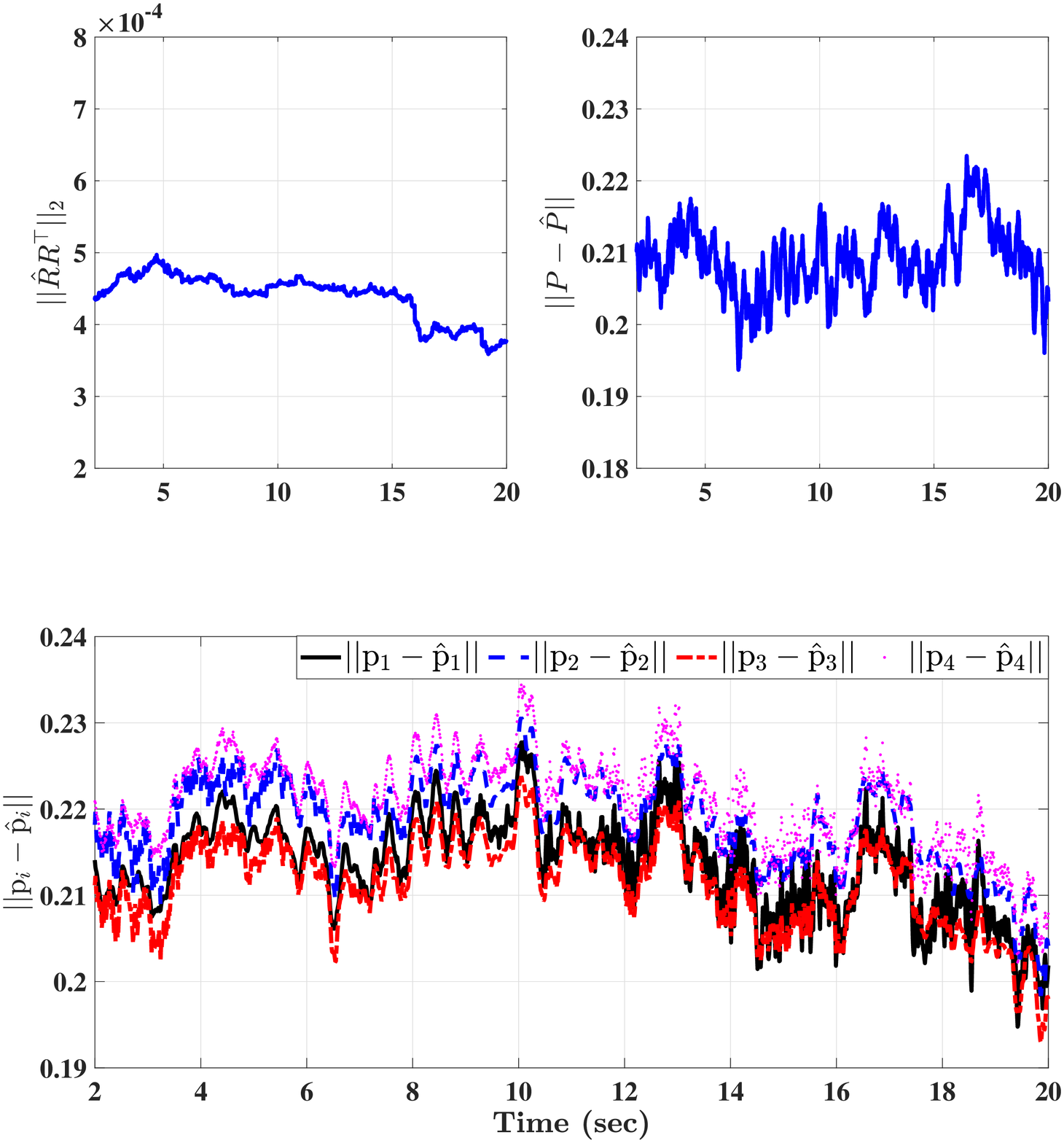}\caption{Steady-state values of $||\hat{R}R^{\top}||_{{\rm I}}$, $||P-\hat{P}||$
		and $||{\rm p}_{i}-{\rm \hat{p}}_{i}||$ for all $i=1,2,3,4$ of the
		proposed nonlinear stochastic filter for SLAM.}
	\label{fig:SLAM_error_Steady_State}
\end{figure}

\subsection{Experimental Validation}

To further validate the proposed nonlinear stochastic estimator for
SLAM, the algorithm has been tested on a real-world EuRoc dataset
\cite{burri2016euroc}. The data set includes 1) the true orientation
and position trajectory of the unmanned aerial vehicle, 2) IMU data,
and 3) stereo images. Due to the fact that the dataset does not include
landmark information, four landmarks fixed with respect to $\left\{ \mathcal{I}\right\} $
have been positioned at ${\rm p}_{1}=[3,0,0]^{\top}$, ${\rm p}_{2}=[-3,0,0]^{\top}$,
${\rm p}_{3}=[0,3,0]^{\top}$, and ${\rm p}_{4}=[0,-3,0]^{\top}$.
The four landmark estimates are initiated at the following positions:
$\hat{{\rm p}}_{1}\left(0\right)=\hat{{\rm p}}_{2}\left(0\right)=\hat{{\rm p}}_{3}\left(0\right)=\hat{{\rm p}}_{4}\left(0\right)=[0,0,0]^{\top}$.
In spite of the large initialization error, Figure \ref{fig:SLAM_3D_DS}
demonstrates smooth and continuous convergence of the robot's position
from the origin to the true trajectory successfully arriving to the
desired destination. Likewise, Figure \ref{fig:SLAM_3D_DS} shows
the convergence of the estimated landmarks from the origin to true
locations.

\begin{figure}[h!]
	\centering{}\includegraphics[scale=0.29]{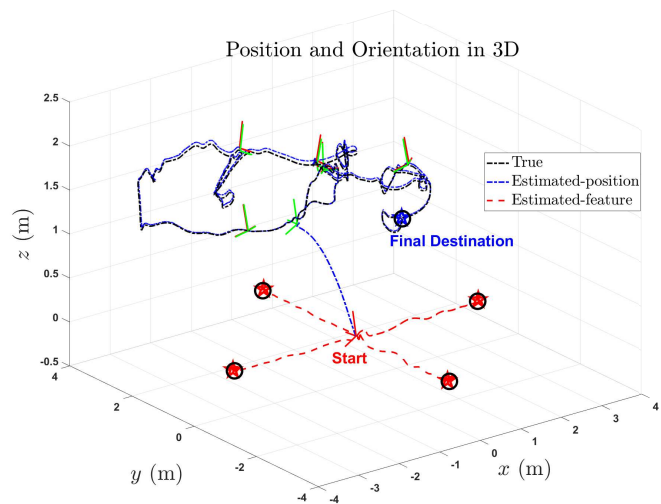}\caption{Experimental results using dataset Vicon Room 2 01.}
	\label{fig:SLAM_3D_DS}
\end{figure}

\section{Conclusion \label{sec:SE3_Conclusion}}

To truly capture the nonlinear structure of the motion dynamics of
Simultaneous Localization and Mapping (SLAM), a nonlinear stochastic
filter for SLAM on the Lie group of $\mathbb{SLAM}_{n}\left(3\right)$
is proposed. The proposed stochastic filter takes into account the
unknown constant bias and random noise corrupting the velocity measurements.
The proposed filter directly incorporates angular and translational
velocity, landmark,  and IMU measurements. The closed loop error signals
have been shown to be semi-globally uniformly ultimately bounded (SGUUB)
in mean square. Numerical results conclusively prove filter's ability
to localize the unknown robot's pose and simultaneously map the unknown
environment.

\section*{Acknowledgment}

The authors would like to thank \textbf{Maria Shaposhnikova} for proofreading
the article.

\bibliographystyle{IEEEtran}
\bibliography{bib_SLAM}

\vspace{170pt}

\section*{AUTHOR INFORMATION}
\vspace{10pt}

	{\bf Hashim A. Hashim} (Member, IEEE) is an Assistant Professor with the Department of Engineering and Applied Science, Thompson Rivers University, Kamloops, British Columbia, Canada. He received the B.Sc. degree in Mechatronics, Department of Mechanical Engineering from Helwan University, Cairo, Egypt, the M.Sc. in Systems and Control Engineering, Department of Systems Engineering from King Fahd University of Petroleum \& Minerals, Dhahran, Saudi Arabia, and the Ph.D. in Robotics and Control, Department of Electrical and Computer Engineering at Western University, Ontario, Canada.\\
	His current research interests include stochastic and deterministic attitude and pose filters, Guidance, navigation and control, simultaneous localization and mapping, control of multi-agent systems, and optimization techniques.

\end{document}